\documentclass[12pt]{amsart}
\DeclareUnicodeCharacter{2212}{-}
 \usepackage{graphicx}
\usepackage{amsmath,amsthm,amsfonts,amscd,amssymb,mathrsfs}
\usepackage{cite}
\usepackage[initials]{amsrefs} 
\usepackage{amscd}   
\usepackage[all]{xy} 
\usepackage{booktabs} 
\usepackage{stmaryrd}
\usepackage{mathtools}
\usepackage{epstopdf}
\usepackage[T1]{fontenc}
\usepackage[utf8]{inputenc}

\usepackage{soul} 
\usepackage{nccmath} 
\usepackage{url}
\usepackage{caption}
\usepackage{subcaption}
\DeclareMathAlphabet{\mathscr}{OT1}{pzc}{m}{it} 

\DeclareMathOperator{\ad}{\textrm{ad}}
\DeclareMathOperator{\tr}{\textrm{trace}}

\usepackage{hyperref}
  \hypersetup{colorlinks=true,citecolor=blue, urlcolor= cyan}
\usepackage{cleveref}
 \AtBeginDocument{%
    \def\MR#1{} 
 }

\usepackage{color}

\newtheorem{theorem}{Theorem}[section]

\def \H{\mathcal{H}}
\newcommand{\QQ}{\mathbb{Q}}
\newcommand{\RR}{\mathbb{R}}
\newcommand{\CC}{\mathbb{C}}
\newcommand{\SL}{\operatorname{SL}}

\theoremstyle{remark}
\newtheorem{remark}[theorem]{Remark}

\usepackage{physics}
\renewcommand{\sl}{\mathfrak{sl}}

\title[Maximally entangled 3-qutrits]{Maximally entangled real states and SLOCC invariants: the 3-qutrit case}
\author{Hamza Jaffali$^{1}$, Fr\'ed\'eric Holweck$^{2,3}$, Luke Oeding$^{3}$}
\date{%
    $^1$ \textit{{ColibrITD}, 91 Rue du Faubourg Saint Honor\'e, 75008 Paris, France}\\
    \href{mailto:hamza.jaffali@colibritd.com}{\textit{hamza.jaffali@colibritd.com}}\\
    ~\\
    $^2$ \textit{{ICB}, UMR 6303, CNRS, University Bourgogne Franche-Comt\'e, {UTBM}, 90010 Belfort, France}\\%
    \href{mailto:frederic.holweck@utbm.fr}{\textit{frederic.holweck@utbm.fr}}\\
    ~\\
    $^3$ \textit{Department of Mathematics and Statistics, Auburn University, Auburn, AL, USA}\\
    \href{mailto:oeding@auburn.edu}{\textit{oeding@auburn.edu}}\\
}

\begin{document}
\maketitle

\begin{abstract}
The absolute values of polynomial SLOCC invariants (which always vanish on separable states) can be seen as measures of entanglement. 
We study the case of real 3-qutrit systems and discover a new set of maximally entangled states (from the point of view of maximizing the hyperdeterminant). We also study the basic fundamental invariants and find real 3-qutrit states that maximize their absolute values. It is notable that the Aharonov state is a simultaneous maximizer for all 3 fundamental invariants. We also study the evaluation of these invariants on random real 3-qutrit systems and analyze their behavior using histograms and level-set plots. Finally, we show how to evaluate these invariants on any 3-qutrit state using basic matrix operations.
\end{abstract}
\section{Introduction}

Quantum entanglement is a fundamental concept in quantum physics that is responsible for the correlation between two or more parties of a multipartite quantum system, independently of the distance separating those components. With no classical equivalent, it is a key resource in Quantum Computing and Quantum Communications and is believed to be one of the main reasons for their speed up and advantage. Being able to classify and quantify quantum entanglement is thus crucial to understand and develop more powerful protocols \cite{JaffaliThesis}.

Qutrit systems, which involve three-level quantum systems, have garnered attention for their potential to enhance the functionality of quantum computing \cite{Gokhale2019}. While entangling qutrits is more challenging than entangling qubits, qutrits offer a richer set of possibilities for entanglement and allow for more complex operations \cites{Jaffali2016,Goss2022}. This makes them promising candidates for a range of quantum applications, including quantum error correction \cite{Gottesman1999} and quantum simulation \cite{Gustafson2022}.

On the other hand, the study of multipartite maximally (or highly) entangled states is relevant in the context of Quantum Information Processing and Quantum Computing. Indeed, maximally entangled states are crucial in the implementation of quantum networks \cite{Perseguers2010}, in the field of MBQC (measurement-based quantum computer) \cite{Raussendorf2001}, in quantum error correcting codes \cites{Gour2007,Raissi2018}, and in quantum communication protocols \cites{Cleve1999,Helwig2013} to mention a few.

For bipartite or multipartite systems (with two or three particles), the question of maximally entangled states has been studied from different angles \cites{Enriquez2016, aulbach2010maximally}. The question becomes more difficult as the number of particles increases (already the four-qubit case is considerably more difficult than the three-qubit case). Depending on the measure we consider, we may obtain radically different states, and the notion of ``maximally entangled'' is not as straightforward as in the two or three-qubit case. In this study, we propose a new candidate from the perspective of algebraic measures of entanglement, more specifically, the absolute value of the hyperdeterminant of $3 \times 3 \times 3$ tensors, representing 3-qutrit quantum systems. We also apply the same methods to the other fundamental SLOCC (Stochastic Local Operations with Classical Communication) invariants and find states that maximize them.

The article is organized as follows. In \Cref{sec:state_art}, we recall the main works concerning 3-qutrits maximally entangled states. In \Cref{sec:hyperdet}, we briefly introduce the notion of the hyperdeterminant and the classification of 3-qutrit systems. In \Cref{sec:state}, we present the new real 3-qutrit maximally entangled state candidate, and we prove that it maximizes the $3 \times 3 \times 3$ hyperdeterminant. In \Cref{sec:invariants} we also propose candidates that maximize the fundamental invariants, i.e., a set of invariant polynomials that can be used to express all algebraic invariants including the $3\times 3\times 3$ hyperdeterminant. \Cref{sec:representation} discusses the distribution of random $3$-qutrit quantum states with respect to the absolute value of the hyperdeterminants and fundamental invariants.  In \Cref{sec:compare} we provide practical tools to evaluate the invariants described in this paper on any $3$-qutrits quantum state. We conclude our work in \Cref{sec:conclusion} and provide directions for future works.

	\subsection{Previous Works}\label{sec:state_art}

The question of maximally entangled for 3-qutrit systems has already been studied in the past. In 2016, Enriquez \textit{et al.} proposed two random states $\ket{\Psi_1}$ and $\ket{\Psi_2}$, and one symmetric 3-qutrit state $\ket{D[3,(1,1,1)]}$ for which the minimal decomposition entropy is the largest \cite{Enriquez2016}:

\begin{equation}\label{eq:psi1}
\begin{split}
    \ket{\Psi_1} = 0.193e^{1.7i}\ket{000} +  0.323e^{−2.01i}\ket{001} + 0.16e^{−2.16i}\ket{002} +  0.229e^{−2.22i}\ket{010} \\ + 0.232e^{−3.12i}\ket{011} + 
    0.186e^{−2.5i}\ket{012} + 0.239e^{−2.34i}\ket{020} + 0.141e^{−0.411i}\ket{021} \\ + 0.159e^{−0.512i}\ket{022}  + 0.099e^{1.54i}\ket{100} + 0.144e^{−2.43i}\ket{101} + 0.148e^{2.13i}\ket{102} \\ + 0.263e^{−1.62i}\ket{110} + 0.322e^{0.475i}\ket{111} + 0.216e^{−1.95i}\ket{112} + 0.068e^{−1.39i}\ket{120} \\ + 0.030e^{−2.89i}\ket{121} + 0.171e^{1.91i}\ket{122}  + 0.253e^{−2.82i}\ket{200} + 0.022e^{−0.225i}\ket{201} \\ + 0.06e^{−1.2i}\ket{202}  + 0.003e^{2.64i}\ket{210} + 0.133e^{−1.52i}\ket{211} + 0.202e^{2.2i}\ket{212} \\ + 0.194e^{1.08i}\ket{220} + 0.207e^{1.13i}\ket{221} + 0.274e^{−2.29i}\ket{222}\;,
\end{split}
\end{equation}

\begin{equation}\label{eq:psi2}
\begin{split}
    \ket{\Psi_2} = 0.245e^{0.074i}\ket{000} + 0.024e^{2.49i}\ket{001} + 0.248e^{1.66i}\ket{002}  + 0.069e^{1.55i}\ket{010} \\ + 0.256e^{0.114i}\ket{011} + 0.118e^{−2.88i}\ket{012}  + 0.313e^{−1.24i}\ket{020} + 0.076e^{2.77i}\ket{021} \\ + 0.149e^{0.208i}\ket{022}  + 0.208e^{2.56i}\ket{100} +  0.227e^{−2.88i}\ket{101} + 0.157e^{2.27i}\ket{102} \\ + 0.072e^{3.08i}\ket{110} + 0.2e^{−1.07i}\ket{111} + 0.199e^{−1.87i}\ket{112}  + 0.13e^{−1.95i}\ket{120} \\ + 0.133e^{1.5i}\ket{121} + 0.218e^{−1.68i}\ket{122} + 0.244e^{−1.84i}\ket{200} +  0.191e^{−3.05i}\ket{201} \\ + 0.049e^{2.61i}\ket{202} + 0.144e^{1.22i}\ket{210} + 0.226e^{2.14i}\ket{211} + 0.278e^{−2.46i}\ket{212} \\ + 0.227e^{0.773i}\ket{220} + 0.186e^{−2.11i}\ket{221} + 0.218e^{−1.52i}\ket{222}\;,
\end{split}
\end{equation}

\begin{equation}\label{eq:D3}
    \ket{D[3,(1,1,1)]} = \frac{1}{\sqrt{6}} \Big( \ket{012} + \ket{021} + \ket{102} + \ket{120} + \ket{201} + \ket{210}  \Big)\;.
\end{equation}

The same year, Hebenstreit \textit{et al.} proposed a way to characterize maximally entangled sets for the generic 3-qutrit states (see \Cref{eq:semi-simple}), and how one can reach them through LOCC (Local Operations and Classical Communication) \cite{Hebenstreit2016}.

Even if the SLOCC classification of three-qutrits is well mathematically understood (see \Cref{sec:nurmiev}), deciding which type of qutrit entanglement is useful for quantum protocols is not straightforward. For instance, in \cite{grandjean2012bell} the violation of Bell-like inequalities is studied for different type of three qutrit entangled states including the generalization of the GHZ-state:
\begin{equation}\label{eq:ghz}
    \ket{\text{GHZ}_{333}}=\frac{1}{\sqrt{3}}\Big(\ket{000}+\ket{111}+\ket{222}\Big)\;,
\end{equation}
two generalizations of the $\ket{W}$ state:
\begin{equation}\label{eq:W}
\begin{array}{rcl}
    \ket{W}&=&\frac{1}{\sqrt{3}}\Big(\ket{100}+\ket{010}+\ket{001}\Big),
    \\
    \ket{W_{333}}&=&\frac{1}{\sqrt{6}}\Big(\ket{100}+\ket{200}+\ket{010}+\ket{020}+\ket{001}+\ket{002}\Big)\;,
    \end{array}
\end{equation}
as well as the Aharonov state or singlet state:
\begin{equation}\label{eq:A}
    \ket{\mathcal{A}}=\frac{1}{\sqrt{6}}\Big(\ket{012}+\ket{201}+\ket{120}-\ket{021}-\ket{102}-\ket{210}\Big)\;.
\end{equation}
In this study, the $\ket{W}$ state achieves higher violation than the $\ket{\text{GHZ}_{333}}$ of Bell's inequality, unlike the three-qubit case. The maximum violation is obtained for a state that was found numerically to be:
\begin{equation}\label{eq:psi3}
    \ket{\psi_3}=\frac{(3-2\sqrt{3})\ket{000}+\ket{011}+\ket{101}+\ket{110}}{2\sqrt{6-3\sqrt{3}}}\;.
\end{equation}
The conclusions of \cite{laskowski2014noise} are similar under the presence of noise, they also considered in their paper the following Dicke states:
\begin{equation}\label{eq:d2d3}
    \begin{array}{c}
    D^2 _3=\frac{1}{\sqrt{15}}\Big(\ket{200}+\ket{020}+\ket{002}+2(\ket{110}+\ket{101}+\ket{011})\Big),\\
    D^3 _3=\frac{1}{\sqrt{10}}\Big(\ket{012}+\ket{021}+\ket{102}+\ket{120}+\ket{201}+\ket{210}+2\ket{111}\Big)\;.
    \end{array}
\end{equation}
Let us also note that the Aharonov state $\ket{\mathcal A}$ was shown to be the resource to solve the quantum byzantine agreement problem \cite{fitzi2001quantum} as well as the liar detection problem \cite{cabello2002n}. We will see that the Aharonov state also maximizes the fundamental invariants for 3 qutrits.

The literature on  entanglement of pure qutrit quantum states is not restricted to the three-partite case and higher system have been discussed \cites{goyeneche2014genuinely,Caves2000,Zha2020}.

\section{The \texorpdfstring{$3\times3\times3$}{} Hyperdeterminant}\label{sec:hyperdet}

The hyperdeterminant, a generalization of the determinant, is a homogeneous polynomial evaluated on tensors. Geometrically,  the zero locus of the hyperdeterminant defines the equation of the dual variety $X^*$, i.e., the (closure) of the set of tangent hyperplanes of $X$, where $X$ is the variety of separable states. The standard reference for hyperdeterminants is \cite{GKZ}, and one may also enjoy reading the introductions by Ottaviani \cites{Ottaviani_Hyperdeterminants, Ottaviani_5Lectures}.

Since the hyperdeterminant is invariant under SLOCC (Stochastic Local Operations and Classical Communication), it can be used to characterize entanglement in quantum systems \cite{Miyake2003}. For example, in the case of three-qubit systems, the hyperdeterminant can be used to differentiate between the two genuine entangled states $\ket{\text{GHZ}}$ and $\ket{W}$ \cite{Holweck2012}. In addition to giving qualitative information about the entanglement of a quantum state, the hyperdeterminant can also be seen as a quantitative measure of entanglement \cites{Jaffali2019, de2021mermin}. The study of hyperdeterminants has led to the development of new techniques for characterizing entangled states and understanding their properties \cites{Miyake2003, Holweck2014, Holweck2014a, Jaffali2016, Jaffali2020, AlsinaThesis, csen2013hyperdeterminants}.

We recall that a general 3-qutrit state $\ket{\psi} \in \H_{333} = \CC^3 \otimes \CC^3 \otimes \CC^3$ can be expressed as:
\begin{equation}
    \ket{\psi} = \sum_{i,j,k=0}^{2} a_{ijk} \ket{i}\otimes\ket{j}\otimes\ket{k}= \sum_{i,j,k=0}^{2} a_{ijk} \ket{ijk}\;, 
\quad \text{with} \quad
    \sum_{i,j,k=0}^{2} |a_{ijk}|^2 = 1\;.
\end{equation}
In the case of three-qutrit systems the SLOCC group is
\[
    G_{SLOCC} = \SL_3(\CC) \times \SL_3(\CC) \times \SL_3(\CC)\;.
\]

Let  $\Delta_{333}$ denote the $3\times 3\times 3$-hyperdeterminant, which is, up to scale, the unique nonzero irreducible polynomial (of degree $36$) in the variables $a_{ijk}$ such that,
\begin{equation}
    \Delta_{333}(\ket \Phi) = 0 \; \Leftrightarrow \;
\exists a,b,c \in \CC^3 \text{ such that }
\left\{ \begin{matrix}
\bra {(a \otimes b \otimes z) }\ket{\Phi} = 0, \\
\bra {(a \otimes y \otimes c) }\ket{\Phi} = 0,\\
\bra {(x \otimes b \otimes c) }\ket{\Phi} = 0
\end{matrix} \right\}
\text{ for all } x, y, z \in \CC^3\;.
\end{equation}

In coordinates, we say that $\Delta_{333}$ is the unique up to scale nonzero polynomial (of degree 36) that vanishes on all tensors that are SLOCC equivalent to $\sum_{I\in \{0,1,2\}^3} x_I\ket{I}$ with 
$x_{000} = x_{001} = x_{002} = x_{010} = x_{020}  = x_{100} = x_{200} = 0$. 
In \Cref{sec:nurmiev}, we recall the classification of 3-qutrits systems under the action of SLOCC. In \Cref{sec:det_oeding}, we recall the expression of the hyperdeterminant for a generic 3-qutrit state. 

    \subsection{Classification of SLOCC Orbits}\label{sec:nurmiev}
\subsubsection{Orbit Classification Over \texorpdfstring{$\CC$}{}}
In 2000, Nurmiev presented a classification of $3 \times 3 \times3$ arrays with entries in $\CC$ up to the action of the group SLOCC \cite{Nurmiev2000}. According to his classification, the orbits of $\H_{333}$ under the action of $G_{SLOCC}$ consist of 5 families (4 depending on parameters, 1 parameter free). The normal forms corresponding to each family are also described in \cite{Nurmiev2000}. 

Crucial to Nurmiev's classification is the embedding of $\H_{333}$ into a graded Lie algebra
\[
\mathfrak{e}_6 \simeq \sl_3^{\times 3} \oplus \H_{333}\oplus \H_{333}^*\;.
\] As such, any tensor $\ket{\psi}$ in $\H_{333}$ can be written in a unique way, as $\ket{\psi} = \ket{\psi_S}+\ket{\psi_N}$ with $\ket{\psi_S}$ semi-simple, $\ket{\psi_N}$ nilpotent part, and $[\ket{\psi_S}, \ket{\psi_N}] = 0$. Note, a tensor of size $3\times3\times3$ is said to be semi-simple if its orbit under SLOCC is closed and said to be nilpotent if the closure of its orbit under SLOCC contains the zero tensor. Moreover, just as the determinant of a matrix only depends on its eigenvalues, it is known that the value of any continuous invariant on $\ket{\psi} = \ket{\psi_S}+\ket{\psi_N} \in \H_{333} $ only depends on $\ket{\psi_S}$ \cite{Bremner2014}.

The semi-simple tensors in the case of 3-qutrits are also called \emph{generic} 3-qutrits states, as it is mentioned in \cite{Hebenstreit2016}. Nurmiev has shown that any complex semi-simple tensor is SLOCC-equivalent to a state of the form
\begin{equation}\label{eq:semi-simple}
\ket{\psi_S} = a \ket{v_1} + b \ket{v_2} + c \ket{v_3},
 \text{ with }  (a,b,c) \in \CC^3 , \text{ and }
    |a|^2 + |b|^2 + |c|^2 = 1\;, \text{ and }
\end{equation}
\begin{equation}\label{eq:basis_3qutrits_semisimple}
\begin{split}
\ket{v_1} =  \frac{1}{\sqrt{3}} \left( \ket{000} + \ket{111} + \ket{222} \right), \\
\ket{v_2} = \frac{1}{\sqrt{3}} \left(\ket{012} + \ket{120} + \ket{201} \right), \\
\ket{v_3} = \frac{1}{\sqrt{3}} \left(\ket{021} + \ket{102} + \ket{210} \right)\;.
\end{split}
\end{equation}

Back to Nurmiev's classification, each family can be defined as a linear combination of the three vectors $\ket{v_1}$, $\ket{v_2}$ and $\ket{v_3}$, plus a nilpotent part. The complex coefficients $a$, $b$ and $c$ associated with these vectors must satisfy a set of conditions, listed as follows:

\begin{itemize}
\item first family $F_1$ : $abc\neq 0$, $\left(a^3+b^3+c^3\right)^3 - \left(3abc\right)^3 \neq 0$,
\item second family $F_2$ : $b\left(a^3+b^3\right) \neq 0$, $c=0$,
\item third family $F_3$ : $a \neq 0$, $b=c=0$, 
\item fourth family $F_4$ : $c=-b \neq 0$, $a=0$.
\end{itemize}

The tensors of the first family correspond to semi-simple tensors and thus have no nilpotent part. We will be especially interested in this family since it is the only one that does not annihilate the hyperdeterminant. The fifth family of Nurmiev's classification, which does not depend on any parameter, is called the nilpotent cone, and only contains nilpotent orbits (no semi-simple part). The variety of nilpotent states coincides with the nilpotent cone, which is the variety where all invariants vanish \cites{Mumford1994,Bremner2013,Bremner2014}.

\subsubsection{Real Semi-simple Elements}\label{sec:real-semisimple}
In \cite{di2023classification}, Di Trani {\em et al.}, proved that the classification of semi-simple elements for real three-qutrit states is like Nurmiev's original proof. Each semi-simple real three-qutrit state belongs to one of the five families. Only the first family has a chance for the hyperdeterminant to not vanish and in \cite{di2023classification} it was shown that all real orbits of the first family have a representative with $a,b,c$ real. In other words, the semi-simple part of any real three-qutrit state that does not vanish the hyperdeterminant is SLOCC equivalent to one the form
\begin{equation}
    \ket{\psi_S}=av_1+bv_2+cv_3 \text{ with } a,b,c\in \RR, \text{ and } a^2+b^2+c^2=1\;.
\end{equation}

According to \cite{di2023classification} the real semisimple elements in the families 2 and 3 split into subcases (and up to permutation), of which we recall their normalized versions:
\begin{itemize}
    \item $F_2'$: Set 
    \[
\ket{w_1} = \sqrt{\frac{2}{21}}\left(-\ket{212} + \ket{200} + 2\ket{120} - 2\ket{111} + \frac{1}{2}\ket{022} - \frac{1}{2}\ket{001}\right)\;,\]
\[   \ket{w_2} = \sqrt{\frac{2}{21}}\left(-\ket{222}- \ket{201}+2\ket{121}+ 2\ket{110}-\frac{1}{2}\ket{012} - \frac{1}{2}\ket{000}\right)\;.
\]
Then real representatives of this type are of the form \begin{equation}\label{eq:ssF2}
    \ket{\psi_{ss,2}} = a_1\ket{w_1} + a_2\ket{w_2}\;,
\end{equation}
for  $a_1,a_2\in \RR$ with $a_1^2+a_2^2= 1$, and $a_2(a_2^2 - 3 a_1^2) \neq 0$.
    \item $F_3'$:  Set  
    \[
    \ket{w} = \frac{5\sqrt{41}}{8}\left(-2\ket{001} - 2\ket{010} - 2\ket{100}+ 2\ket{111}-\frac{1}{8}\ket{222}\right)\;.
    \]
Then the elements of this family have the form $aw$ for $a\in \RR$. To have a state (with norm 1) we require that $a^2 = 1$, so the real states have $a=\pm 1$. 
\begin{equation}\label{eq:ssF3}
\ket{\psi_{ss,3}} = \pm\frac{5\sqrt{41}}{8}\left(-2\ket{001} - 2\ket{010} - 2\ket{100}+ 2\ket{111}-\frac{1}{8}\ket{222}\right)\;.
\end{equation}

\end{itemize}

    \subsection{Expressing the \texorpdfstring{$3\times 3\times 3$}{} Hyperdeterminant}\label{sec:det_oeding}
It is known that Schl\"afli's method computes $\Delta_{333}$, but this is one of the few cases for which such a straightforward method exists \cite{WeymanZelevinsky_sing}. 
On the other hand, in this case the invariant ring for $G_{SLOCC}$ is finitely generated by three fundamental invariants $I_6$, $I_9$, and $I_{12}$ (respectively of degrees 6, 9 and 12), whose expressions are found in \cites{Briand2004, Bremner2013}, and we recall these in \Cref{sec:invariants}.
 In 2014, Bremner, Hu, and Oeding \cite{Bremner2014} provided an expression of the hyperdeterminant $\Delta_{333}$ as a polynomial in the fundamental invariants:
\begin{equation}\label{eq:d333}
    \Delta_{333} = {I_6}^3 {I_9}^2 - {I_6}^2 {I_{12}}^2 + 36 {I_6} {I_9}^2 {I_{12}} + 108 {I_9}^4 - 32 {I_{12}}^3\;.
\end{equation}
 The restriction of \Cref{eq:d333} to normalized states  $\ket{\psi_S}$ from \Cref{eq:semi-simple} is
\begin{equation}\label{eq:hyperdet}
\begin{array}{rcl}
\Delta_{333}(\ket{\psi_S})&=&-\frac{4}{3^{18}}\,{a}^{3}{b}^{3}{c}^{3} \left( a+b+c \right) ^{3} \times 
\\ && \left( {a}^{2}+2\, ab-ac+{b}^{2}-bc+{c}^{2} \right) ^{3}  \left( {a}^{2}-ab+2\,ac+{b}^{2}-bc+{c}^{2} \right) ^{3} \times 
\\ && \left( {a}^{2}-ab-ac+{b}^{2}+2\,bc+{c}^{2}  \right) ^{3}  \left( {a}^{2}-ab-ac+{b}^{2}-bc+{c}^{2} \right)^{3}\;.
 \end{array}
\end{equation}

\section{New Maximally Entangled 3-qutrit States}\label{sec:state}

In this section, we present a set of real 3-qutrit states that are maximally entangled from the perspective of the hyperdeterminant. We determine the maximum of the absolute value of the hyperdeterminant for 3-qutrit real states by approximate numerical methods, and then establish an exact and simplified expression of the real state maximizing it via exact symbolic methods.

We first optimize the absolute value of the hyperdeterminant, using numerical methods (gradient-based, meta-heuristic (random walk, particle swarm optimization)), for both real and complex states. We obtained the same assumed maximum, of about $6.907059 \times 10^{-13}$. For example, here are two states $\ket{\psi_{S_1}}$ and $\ket{\psi_{S_2}}$, reaching the assumed maximal numerical value of the hyperdeterminant, expressed as:

\begin{equation}
\ket{\psi_{S_1}} :  (-0.4597089177,\; 0.6279551660, \;0.6279649847)\;,
\end{equation}
\begin{equation}
\begin{split}
\ket{\psi_{S_2}} : (0.4187234964+0.1897453668 i,\; -0.5719715278-0.2591902230 i,\\ -0.5719688559-0.2592064016 i)\;.
\end{split}
\end{equation}

We studied the real critical points of the hyperdeterminant using the software \verb|Maple| and \verb|Mathematica|, and repeated the computation in \texttt{Macaulay2} \cite{M2} to provide an exact expression of the maximum as well as completely enumerate the maximizers:

\begin{theorem}\label{theo:max_hyperdet_333}
The global maximum of the absolute value of the hyperdeterminant $|\Delta_{333}|$, when restricted to real states is $\frac{\sqrt{3}}{2^{19}\times 3^{14}}$. The global max is reached at 12 semi-simple points $a\ket{v_1} + b\ket {v_2} + c\ket{v_3}$ with the following values and their permutations:
 \begin{equation}
(a,b,c) = \left(rs,s,s\right)
,\quad \text{with} \quad r = (1\pm \sqrt 3)\quad \text{and} \quad 
s=\pm \sqrt{\frac{1}{(r^2 +  2)}}\;.
\end{equation}
\end{theorem}

\begin{proof}
Let $f=\Delta_{333}$ be the $3\times3\times3$ hyperdeterminant, (see \Cref{eq:hyperdet}). Let $a,b,c \in \RR$ such that $a^2 + b^2 + c^2 = 1$. 

We aim to determine all critical points of the hyperdeterminant, i.e., points where the gradient of $f$ is zero. We can rewrite $c=\pm \sqrt{1-a^2-b^2}$, in order to simplify the study of $f$ to the two variables $a$ and $b$. We distinguish three different cases: 

\noindent\textbf{Case 1:} $c=0$ or $c=1$

    The factors $a^3$, $b^3$ and $c^3$ appear in the expression of $f$, and therefore if $c=0$, $f(a,b,c)=f(a,b,0)=0$. If $c=1$, then $a=b=0$, and we also have $f(a,b,c)=f(0,0,c)=0$. In that case, $|f|$ reaches the minimum. 

\noindent\textbf{Case 2:} $c=\sqrt{1-a^2-b^2}$

    If we substitute the expression of $c$ in $f$, we retrieve:

    \begin{equation}
    \begin{split}
        f(a,b) = -{\frac {4}{3^{18}}}\,{a}^{3}{b}^{3} \left( -{a}^{2}-{b}^{2}+1 \right) ^{3/2} \left( a+b+\sqrt {-{a}^{2}-{b}^{2}+1} \right) ^{3}\\ 
        \left( 2\,ab-a\sqrt {-{a}^{2}-{b}^{2}+1}-b\sqrt {-{a}^{2}-{b}^{2}+1}+ 1 \right) ^{3} \\ 
        \left( -ab+2\,a\sqrt {-{a}^{2}-{b}^{2}+1}-b\sqrt {-{a}^{2}-{b}^{2}+1}+1 \right) ^{3} \\
        \left( -ab-a\sqrt {-{a}^{2}-{b}^{2}+1}+2\,b\sqrt {-{a}^{2}-{b}^{2}+1}+1 \right) ^{3} \\
        \left( -ab-a\sqrt {-{a}^{2}-{b}^{2}+1}-b\sqrt {-{a}^{2}-{b}^{2}+1}+1 \right) ^{3}\;.
    \end{split}
    \end{equation}

    If we differentiate $f$ in both variables and we solve the system of equation $\{ \frac{\partial f(a,b)}{\partial a}= 0, \frac{\partial f(a,b)}{\partial b}= 0\}$ over the real numbers. We use symbolic calculus for solving these equations, namely the \verb|Maple| function \verb|RealDomain[solve]|. We obtain 40 different solutions, on which we evaluate $f(a,b)$. Some of the solutions are points in $\RR^2$, and some depend on $b$. For the points we evaluate the function and take the absolute value of the result. 
    
    For the solutions depending on $b$, we substitute them in the expression of the hyperdeterminant to retrieve an expression $f(b)$ only depending on one variable. We again differentiate this expression and search for critical points. We obtain new points, on which we also evaluate the hyperdeterminant. We regroup all values and look for the maximal one, in absolute value. In this case, the maximum value is exactly $\frac{\sqrt{3}}{2^{19}3^{14}}$.

\noindent\textbf{Case 3:} $c=-\sqrt{1-a^2-b^2}$.

This case is entirely analogous to Case 2, and the results turn out to be identical.
%
   
Hence the global maximum of $|\Delta_{333}|$ is exactly equal to $\frac{\sqrt{3}}{2^{19}3^{14}}$. Moreover, the maximum values occur at the points listed in the statement of the theorem.

\noindent\textbf{Second proof and enumeration of maximizers:}

We also ran the same procedure as outlined in the proof of \Cref{thm:otherInvariants} to algebraically compute the set of critical points. Set $q  =  a^2 + b^2 + c^2 -1$, the equation of the sphere over the real numbers. We found 60 critical points on 0-dimensional components defined by the system of equations $\{\nabla \Delta_{333}, q \}$, 36 of which were real solutions. Among the real solutions, the maximum values of $|\Delta_{333}|$ are obtained at 3 ideals (symmetric up to permuting $a,b,c$), one of which is:
$\langle b-c,\,a^{2}-2\,a\,c-2\,c^{2}, a^2 + b^2 + c^2 -1\rangle$, and which has the following 4 solutions:
\[
(a,b,c) = \left(rs,s,s\right)
,\quad \text{with} \quad r = (1\pm \sqrt 3)\quad \text{and} \quad 
s=\pm \sqrt{\frac{1}{(r^2 +  2)}}.
\]
The approximations of the coordinates of the solutions in this case are:
\[\begin{array}{rr}
(.459701, -.627963, -.627963)\;, &
(.888074, .325058, .325058)\;, \\
(-.459701, .627963, .627963)\;, &
(-.888074, -.325058, -.325058)
\;.
\end{array}\]
Permuting $a,b,c$ we obtain the 12 states that attain global max of $|\Delta_{3,3,3}|$. 

Finally, we note that the other types of semi-simple elements \eqref{eq:ssF2} and \eqref{eq:ssF3} vanish identically on $\Delta_{333}$, so cannot be global maxima.
\end{proof}

\begin{remark}
We conjecture that the maximum value of the hyperdeterminant for real three-qutrit states is in fact also the maximum value on complex values. Indeed, we ran several heuristic optimization methods over the complex for a generic semi-simple element $\ket{\psi_S}=av_1+bv_2+cv_2$ with $(a,b,c)\in \CC^3$ and $|a|^2+|b|^2+|c|^2=1$. All our methods returned the same maxima. Similarly, perturbing over the complex numbers the extremum found over the real does not provide any better solution.  
\end{remark}

\section{Maximizing the Fundamental Invariants}\label{sec:invariants}

For this section, we will use the following expressions of the fundamental invariants (see  \cite{Bremner2014}) evaluated on a general semi-simple normalized state $\ket{\psi_{S}}$ as expressed in \Cref{eq:semi-simple}:

\begin{equation}
   I_{6}(\psi_{S}) = \frac{1}{27}\left( a^6-10a^3b^3-10a^3c^3+b^6-10b^3c^3+c^6 \right)\;,
\end{equation}

\begin{equation}
   I_{9}(\psi_{S}) = \frac{-\sqrt{3}}{243} \left(a-b\right) \left(a-c\right) \left(b-c\right) \left(a^2 + ab+ b^2\right) \left(a^2 + ac + c^2\right) \left(b^2 + bc + c^2\right)\;,
\end{equation}

\begin{equation}I_{12}(\psi_{S}) = \frac{1}{729}\left(\begin{array}{l}
a^{9}b^{3}+a^{3}b^{9}+a^{9}c^{3}+b^{9}c^{3}+a^{3}c^{9}+b^{3}c^{9} \\
\hspace{1cm}-4\,(a^{6}b^{6}+a^{6}c^{6}+b^{6}c^{6})
+2\,(a^{6}b^{3}c^{3}+a^{3}b^{6}c^{3}+a^{3}b^{3}c^{6})
\end{array}
\right)\;.
\end{equation}
We note that the basic invariant in degree 12 is only defined up to a  multiple of $I_6^2$, but one might argue that this version of $I_{12}$ is natural because it does not contain the monomial $a^{12}$, for instance, and has small integer coefficients (up to a global rescaling).

\begin{theorem}\label{thm:otherInvariants}
Let $\ket{v_1}$, $\ket{v_2}$ be the first two basic semi-simple states defined at \eqref{eq:basis_3qutrits_semisimple}.
Then, the global maximum of the absolute value of the fundamental invariants $I_{6}$, $I_{9}$ and $I_{12}$ restricted on generic 3-qutrits, with maximum values  respectively $\frac{1}{18} = .05555$, $\frac{\sqrt{6}}{3888} = 0.0006300127939$ and $\frac{1}{7776} = 0.000128600823$, is reached for the real 3-qutrit state
\[
\ket{M_{333,I}} = \frac{1}{\sqrt{2}} \Big( \ket{v_1} - \ket{v_2} \Big) \;.
\]
Moreover, all real semi-simple states that obtain the maximum are permutations of $\ket{M_{333,I}}$. 
\end{theorem}
\begin{proof}
Here is an algebraic method for finding real critical points of a polynomial $f(a,b,c)$ on the sphere $\{a^2 + b^2 + c^2 =1\}$, which we carried out in Macaulay2 (\cite{M2}) but could be done in any symbolic algebra system. 

Since the variables $a,b,c$ are not independent on the sphere, we could solve for $c$ and treat $a,b$ as independent variables as was done above. This involves square roots, which can cause problems with some symbolic methods. Instead, we implicitly differentiate and compute the gradient as follows:
First implicitly differentiate the equation $a^2 + b^2 + c^2 =1$ with respect to $a$ and $b$ treating $a,b$ as independent and $c = c(a,b)$ in order to solve for the values 
\[
\frac{\partial c}{\partial a} = \frac{-a}{c} \quad\text{and}\quad \frac{\partial c}{\partial b} = \frac{-b}{c}\;.
\]
Then compute
\[
\nabla f(a,b,c) = \left(\frac{\partial f}{\partial a}  + \frac{\partial f}{\partial c}\frac{\partial c}{\partial a},
\frac{\partial f}{\partial b}  + \frac{\partial f}{\partial c}\frac{\partial b}{\partial a}
\right)\;.
\]
Set $q = a^2 + b^2 + c^2 -1$. Now we seek to solve the ideal $\{q, \nabla f\}$. We do this both symbolically and numerically. Symbolically we perform an irreducible decomposition of the ideal (the command \texttt{decompose} in \texttt{M2}) and obtain a finite list of ideals $J_i$ that are irreducible over $R = \QQ[a,b,c]$.
It is typical when working on a computer over an exact field like the rationals $\QQ$ that some solutions will not exist over this field, which is why we also compute the numerical decomposition. Still, with a component of our ideal that is not solvable over $\QQ$ we may still be able to evaluate $|f|$ at the points $\mathcal{V}(J_i)$ by substituting $f$ into each ring $R/J_i$ to attempt to obtain a value for $f$ at the points in the zero-set $\mathcal{V}(J_i)$. The result will either be a number, or an ideal (typically using fewer variables) that we can attempt to solve. 

As a caution, we are only interested in the values of $f$, so we do not consider the values of $f$ at complex critical points $(a,b,c)$. Even more, the polynomial $q = a^2+b^2+c^2-1$ only vanishes for real points $(a,b,c)$ of norm 1, the other non-real complex roots will not have norm 1 and will not be solutions to our system of equations. 

A numerical irreducible decomposition may be computed via the command \texttt{solveSystem} using the Numerical Algebraic Geometry package \cite{NAG} in \texttt{M2}. This will produce a list $L$ numerical solutions to the ideal $\{q,\nabla f\}$. Then we can evaluate $f$ at every point of $L$. 

We carried out this procedure in Macaulay2 for each invariant $f \in \{I_6, I_9, I_{12}\}$, and compared the values of $|f|$ at critical points.  The maximal values occur precisely at the critical points listed in the statement of the theorem. We include this computation in the ancillary files accompanying the arXiv version of this article. 
\begin{remark}\label{rem:critPts}
Some of the critical points coincided with the states $\ket{D[3,(1,1,1)]}$ at \eqref{eq:D3} and $\ket{\text{GHZ}_{333}}$, but they attain value $|I_6(\ket{D[3,(1,1,1)]})|=|I_6(\ket{\text{GHZ}_{333}})|= \frac{1}{27}$, which is smaller than the global maximum. The global max is obtained for the Aharonov state $\ket{\mathcal A}$ \eqref{eq:A}.  
\end{remark}

We also compared the values of $|f|$ to the critical values found on the other real semi-simple elements (see \Cref{sec:real-semisimple}) for each respective invariant. We found that none of them were greater than the maxima we found on the generic semi-simple elements.

Here is the argument in more detail.
The general semi-simple element $\ket{\psi_{S}}$ degenerates to each of the other families of semi-simple elements over the complex numbers. However, over the real numbers two additional orbits must be considered corresponding to the states \eqref{eq:ssF2} and \eqref{eq:ssF3}, \cite{di2023classification}. 
Using the methods in \Cref{sec:representation} we can also compute these invariants on the other families of real semi-simple elements, which we denote respectively by $\ket{\psi_{ss,2}}$ and $\ket{\psi_{ss,3}}$.
Set normalization factors $\zeta_6 = \frac{2^5}{3^3\cdot7^3} $, $\zeta_9 = \frac{2^5\sqrt{42}}{3^7\cdot 5^7}$, and $\zeta_{12} = \frac{2^7}{21^6}$.
We calculated:
\[\renewcommand{\arraystretch}{1.5} \begin{array}{rcl}
I_6(\ket{\psi_{ss,2}}) &=& 
\zeta_6\;(3a_{1}^{6}+15a_{1}^{2}a_{2}^{4}+2a_{2}^{6})
\\ &\equiv& \zeta_6 \;(16\,a_{1}^{6}-24\,a_{1}^{4}+9\,a_{1}^{2}+2) \mod{\langle a_1^2+a_2^2-1\rangle}
, 
\\[2ex]
I_9(\ket{\psi_{ss,2}}) &=& 
\zeta_9
\;\left( -a_{1}^{9}+6\,a_{1}^{5}a_{2}^{4}+8\,a_{1}^{3}a_{2}^{6}+3\,a_{1}a_{2}^{8} \right)
\\ &\equiv& \zeta_9
\;(
-4\,a_{1}^{3}+3\,a_{1}) \mod{\langle a_1^2+a_2^2-1\rangle}
,\\[2ex]
I_{12}(\ket{\psi_{ss,2}}) &=& \zeta_{12}\;\left(-3\,a_{1}^{12}-3\,a_{1}^{8}a_{2}^{4}-40\,a_{1}^{6}a_{2}^{6}-57\,a_{1}^{4}a_{2}^{8}-24\,a_{1}^{2}a_{2}^{10}-\,a_{2}^{12} \right)
\\ & \equiv &\zeta_{12}\;
(-32\,a_{1}^{6}+48\,a_{1}^{4}-18\,a_{1}^{2}-1) \mod{\langle a_1^2+a_2^2-1\rangle},
\\[2ex]
 \Delta_{333}(\ket{\psi_{ss,2}}) &=& 0\;.
\end{array}
\]
Note the second case for each invariant we have simplified the expression and replaced all instances $a_2$ via $a_1^2+a_2^2=1$, which does not involve any square roots since each invariant has even degree in $a_2$.
Since the resulting reduced expressions only involve $a_1$ we can use basic 1-variable optimization from calculus. Here is a summary of the critical points and values:
\[
\begin{array}{c|c|c|c|c}
     &  (0,1) &  (\pm \frac{1}{2}, \pm \sqrt{\frac{3}{4}}), & (\pm \sqrt{\frac{3}{4}}, \pm \frac{1}{2}) & \text{max}|I_d|\\ \hline
I_6: & 2\cdot \zeta_6  & 3\cdot \zeta_6 & 2\cdot \zeta_6 &
3\cdot \zeta_6  = .0103660511823777
\\ \hline
I_9: &  & \pm 1 \cdot \zeta_9 &  & \zeta_9 =.00000121376835394049 \\ \hline
I_{12}: & -1\cdot \zeta_{12} & -3\cdot \zeta_{12} & -1\cdot \zeta_{12} & 3\zeta_{12} =  .00000447729237981977 \\ \hline
\end{array}
\]
Note that all the absolute values of these invariants take their maximum value at the same points, $(a_1,a_2) = (\pm \frac{1}{2}, \pm \sqrt{\frac{3}{4}})$, however, these maxima are all smaller than the values for the other type of semi-simple elements.

In the case of $\ket{\psi_{ss,3}}$ only $I_6$ is non-zero, and since $a^2 = 1$ we have
\begin{equation}
|I_6(\ket{\psi_{ss,3}})| = |-\frac{2^{18}}{5^6\cdot 41^3}a^{6}|
 = \frac{2^{18}}{5^6\cdot 41^3} = .000243426763976147\;,
\end{equation}
which is smaller than the maximum value of $\frac{1}{18}$.
\end{proof}

\begin{remark}
As mentioned already the space of degree $12$ invariants is spanned by $I_6^2$ and $I_{12}$, so it doesn't make sense geometrically to choose this particular $I_{12}$ over any other degree 12 invariant, i.e. we can add in any multiple of $I_6^2$ and replace $I_{12}$ with $I_{12} + \lambda I_6$. 
Then we could re-interpret the maximum value computation as follows. 
Maximizing a single invariant $f$ is finding the unit vector $x$ such that $|f(x)|$ is greatest. Another interpretation of this, which generalizes, is to find the unit vector $x$ that is furthest from the variety $\mathcal{V}(f) = \{y \mid f(y) = 0\}$. So, we replace the question of maximizing $I_{12}$ by the question of finding the point(s) of maximal distance from the variety $\mathcal{V}(I_6, I_{12})$, since this variety eliminates the ambiguity in the definition of $I_{12}$.  However, working on semi-simple elements (defined by $3$ real parameters $a,b,c$), and intersecting with the sphere $a^2 + b^2 + c^2 = 1$ we obtain a finite collection of pairs of antipodal points, and as such every point outside that set maximizes the sum of the distances to these points.
\end{remark}

\section{Evaluating Invariants on Random States}\label{sec:representation}
In this section, we study the evaluation of the fundamental invariants, as well as the hyperdeterminant, on random real 3-qutrit states. We provide several visualizations of these evaluations and try to extract useful information concerning the distribution of the amount of entanglement (in the sense of the absolute value of invariants) in the space of real 3-qutrits.
By the term ``random'' we mean we have used a random number generator with uniform distribution on the interval $[-1,1]$.
Moreover, since we are only concerned with the semi-simple part (the nilpotent parts do not change the values of these invariants) we sample uniformly $a$, $b$ and $c$ in $[-1,1]$, and then renormalize. 

For each invariant, we propose three types of plots. The first type is a histogram representing the distribution of values of the absolute value of the invariant. It provides information about which values of the invariant are most present in the space of real 3-qutrit states. The histograms of the hyperdeterminant $\Delta_{333}$ and the invariants $I_6$, $I_9$ and $I_{12}$ (respectively \Cref{fig:hyperdet_histo,fig:i6_histo,fig:i9_histo,fig:i12_histo}) present shapes that are different from the one random homogeneous polynomials of the same degree can provide. For this first plot, we use 500000 random real 3-qutrits as our input data.

We note that the chances to get (very close to) a maximally entangled state based on the number of samples landing in the last bin in the histogram are as follows. For $\Delta_{333}$, it is $0.3294 \%$. For $I_6$, it is  $0.5802\%$. For $I_9$, it is $0.4122 \%$. For $I_{12}$, it is $0.2696 \%$. For maximizing all invariants at the same time, it is $0.3106 \%$.

For the second type of plot, we follow the same approach as Alsina in his thesis \cite{AlsinaThesis}, by generating a significant number of random states, then ordering them in increasing order to obtain a graphical representation. In this second type of plot, we use 20000 random real 3-qutrits as our input data. See \Cref{fig:i6_curve,fig:i9_curve,fig:i12_curve,fig:hyperdet_curve}.

The last type of figure represents the evaluation of the invariant in a spherical plot, generated in \texttt{Maple}. Because we restrict to semi-simple states with real coefficients, we can represent them in the unit sphere defined by $a^2 + b^2 + c^2 = 1$. Depending on the value of the invariant (we don't compute the absolute value here), the color of a given point indicates if we are close to a minimum/maximum or to zero.
See \Cref{fig:hyperdet_sphere,fig:i6_sphere,fig:i9_sphere,fig:i12_sphere}.

    \subsection{Random States and \texorpdfstring{$|\Delta_{333}|$}{}}
Now we provide some observations and comments on the plots. 
In the histogram \Cref{fig:hyperdet_histo} and curve \Cref{fig:hyperdet_curve} one can see that most of the states have value of hyperdeterminant close to 0. In addition, it appears that the states close to the maximum value are the less probable (according to the histogram).

Looking at the plot on the sphere \Cref{fig:hyperdet_sphere} one can observe the symmetry of the hyperdeterminant, since the expression of the hyperdeterminant is invariant under permutations of $a, b, c$. We observe that the dominant color is the cyan, expressing the fact that most of the real random states have hyperdeterminant value close to 0, echoing what was noticed from the histogram. 

One can guess from the image that the midpoint between two red points is (most of the time) falling either into the cyan or the red color. Which might mean that the mean of two max states could annihilate the hyperdeterminant or give a another max (in the opposite sign). We saw this in the results from \Cref{theo:max_hyperdet_333} that there are 12 points that realize the maximum, and their coordinates differ by permutations, and sign changes.

It seems that the barycenter of three minimally entangled states (red, in the green/yellow zone) and maximally entangled states (red, in the purple blue zone) is the same: a root of the hyperdeterminant.

\begin{figure}[h]
    \centering
    \begin{subfigure}[b]{0.49\textwidth}
         \centering
         \includegraphics[width=\textwidth]{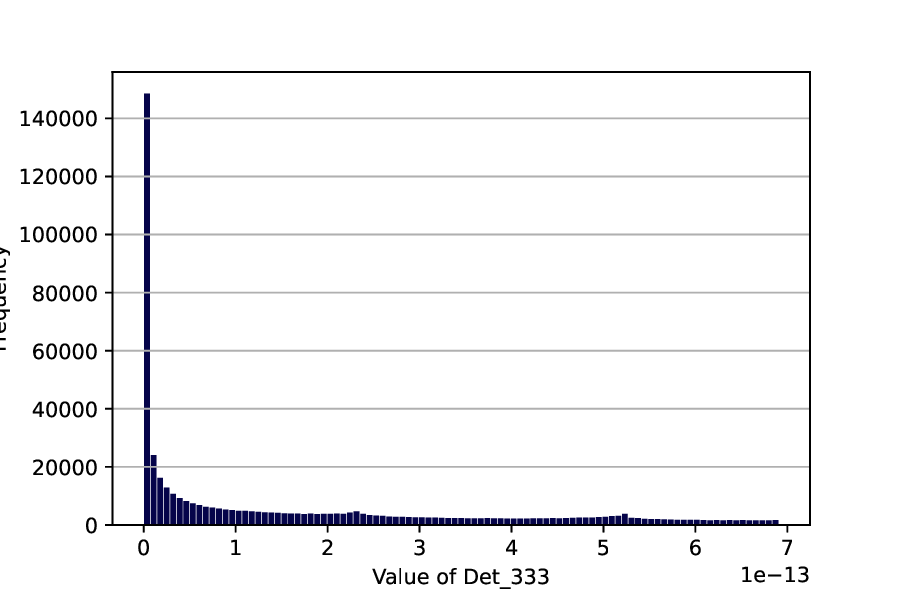}
         \caption{Histogram of the absolute value of the hyperdeterminant for 500,000 randomly generated generic real 3-qutrit states}
         \label{fig:hyperdet_histo}
     \end{subfigure}
     \hfill
    \begin{subfigure}[b]{0.49\textwidth}
         \centering
         \includegraphics[width=\textwidth]{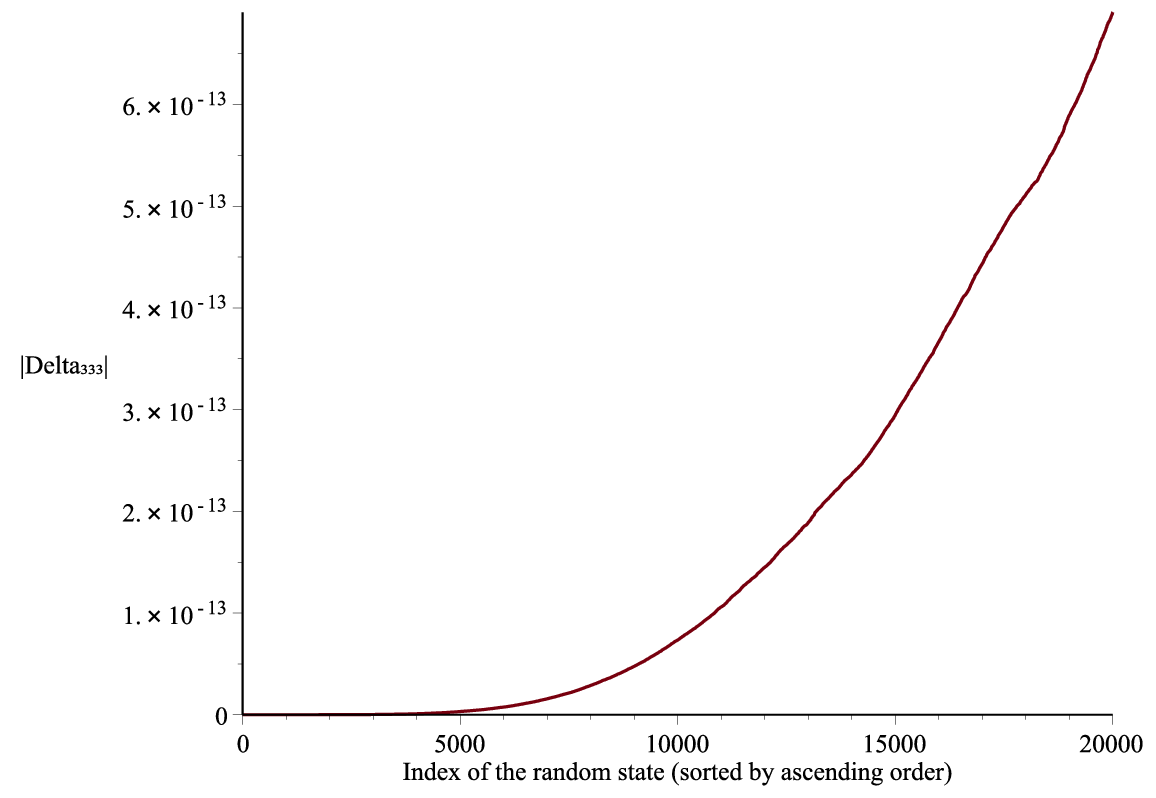}
         \caption{Absolute value of the hyperdeterminant for 20,000 randomly generated generic real 3-qutrit states, sorted in ascending order}
         \label{fig:hyperdet_curve}
     \end{subfigure}
    \caption{}
    \label{fig:hyperdet_plots}
\end{figure}

\begin{figure}[h]
    
\begin{subfigure}[b]{0.49\textwidth}\label{fig:i6_sphere}
    \centering 
    \caption{}    \includegraphics[width=0.7\textwidth]
    {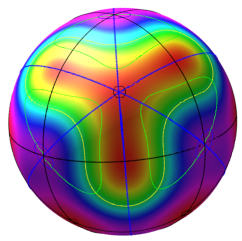}
\end{subfigure}
\begin{subfigure}[b]{0.49\textwidth}\label{fig:i9_sphere}
    \centering
    \caption{} \includegraphics[width=0.7\textwidth]{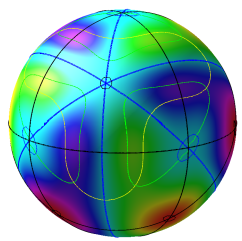}
\end{subfigure}
\begin{subfigure}[b]{0.49\textwidth}\label{fig:i12_sphere}
    \centering
    \caption{} \includegraphics[width=0.7\textwidth]{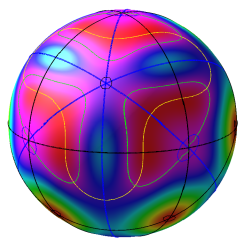}
\end{subfigure}
\begin{subfigure}[b]{0.49\textwidth}\label{fig:hyperdet_sphere}
    \centering
    \caption{} \includegraphics[width=0.7\textwidth]{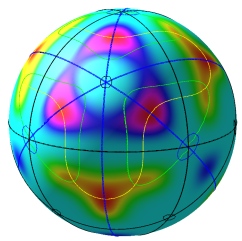}
\end{subfigure}

~

 \begin{tabular}{c|c|c|c|c|c|c}
          Red & Purple & Dark Blue & Cyan & Green & Yellow & Red  \\
          \hline 
          Maximum & Positive & Positive & Close to zero & Negative & Negative & Minimum 
    \end{tabular}
    \caption{A plot of the level sets of the invariants $I_6$ (top left, (A)), $I_9$ (top right, (B)), $I_{12}$ (bottom left, (C)), $\Delta_{333}$ (bottom right, (D)) on the sphere of semi-simple states. The solid curves represent the zero-sets of the invariants, the small circles indicate multiple roots. 
    The line colors of Yellow, Blue, Green, Black respectively correspond to the level sets
         $I_6=0$, $I_9=0$, $I_{12}=0$, $\Delta_{333} = 0$. The other point values are according to the color scheme above.         
}
\end{figure}
    \subsection{Random States and \texorpdfstring{$|I_6|$}{}}

Here are some observations regarding the plots in \Cref{fig:i6_plots}. 
It appears that the values of the invariant have the same frequency (around 2700 in the histogram \Cref{fig:i6_histo}), except around 0.02 and 0.04, where we can observe a kind of Gaussian with a high peak (above 14000). This feature is also seen from  \Cref{fig:i6_curve}, where the curve is a linear function until around 10000 then if almost constant between 10000 and 15000 and then become again a linear function from 15000 to 20000.

We speculate that the high frequency of states having  almost the same value of this invariant can tell us something about the size of some SLOCC orbits.

In the sphere plot \Cref{fig:i6_sphere} we notice the symmetry in the plot due to the invariance under permuting $a,b,c$. We also notice some concurrences between roots of one invariant and peeks in the values of this invariant.

\begin{figure}[h]
    \centering
    \begin{subfigure}[b]{0.49\textwidth}
         \centering
         \includegraphics[width=\textwidth]{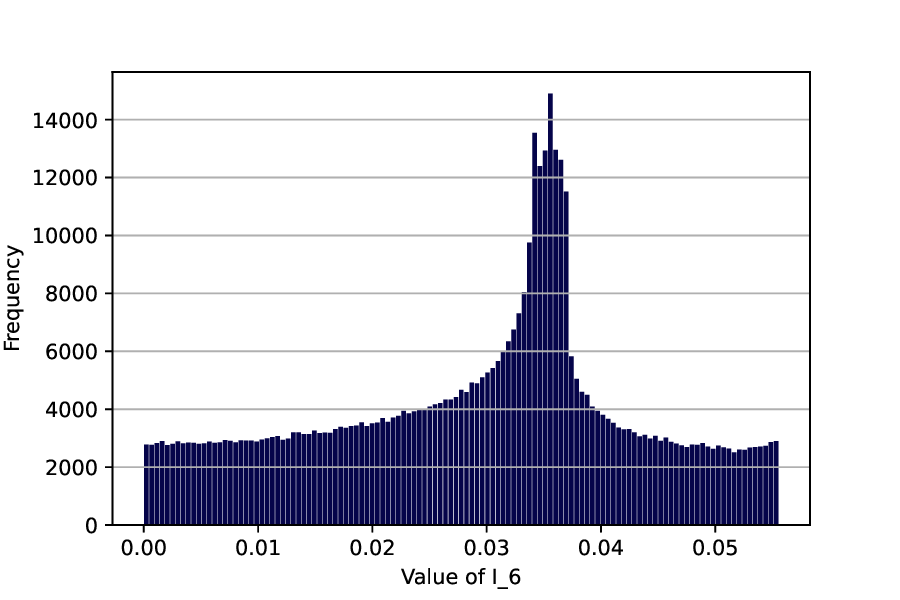}
         \caption{Histogram of the absolute value of fundamental invariant $I_{6}$ for 500,000 randomly generated generic real 3-qutrit states}
         \label{fig:i6_histo}
     \end{subfigure}
     \hfill
    \begin{subfigure}[b]{0.49\textwidth}
         \centering
         \includegraphics[width=\textwidth]{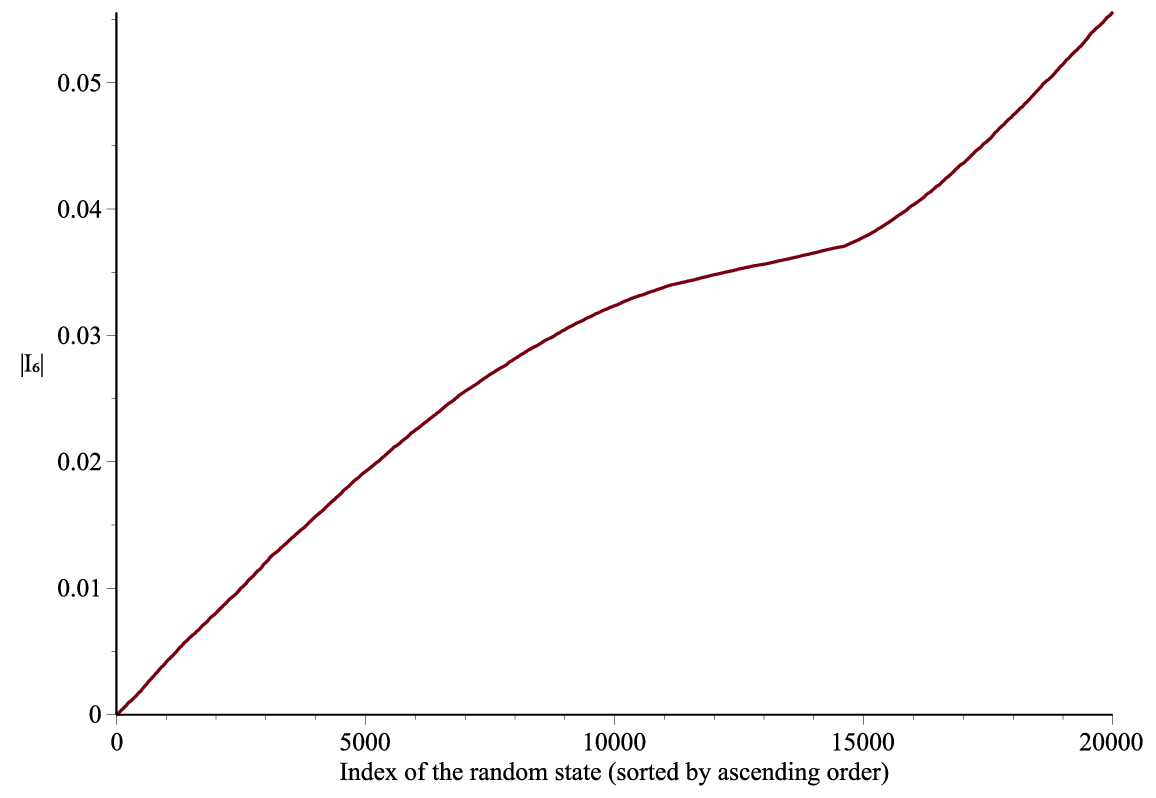}
         \caption{Absolute value of fundamental invariant $I_{6}$ for 20,000 randomly generated generic real 3-qutrit states, sorted in ascending order}
         \label{fig:i6_curve}
     \end{subfigure}
    \caption{}
    \label{fig:i6_plots}
\end{figure}

    \subsection{Random States and \texorpdfstring{$|I_9|$}{}}
The histogram \Cref{fig:i9_histo} and curve \Cref{fig:i9_curve} indicate that a large number of states have value close to zero, and there is another small peak in the histogram at around .000275, indicating another region where $I_9$ has more frequent values. In the sphere plot we also note the $a,b,c$ permutation symmetry, and we notice the alternating nature (opposite signs on adjacent sides of the level curve $I_9=0$) due to the fact that the invariant is of odd degree.
\begin{figure}[h]
    \centering
    \begin{subfigure}[b]{0.49\textwidth}
         \centering
         \includegraphics[width=\textwidth]{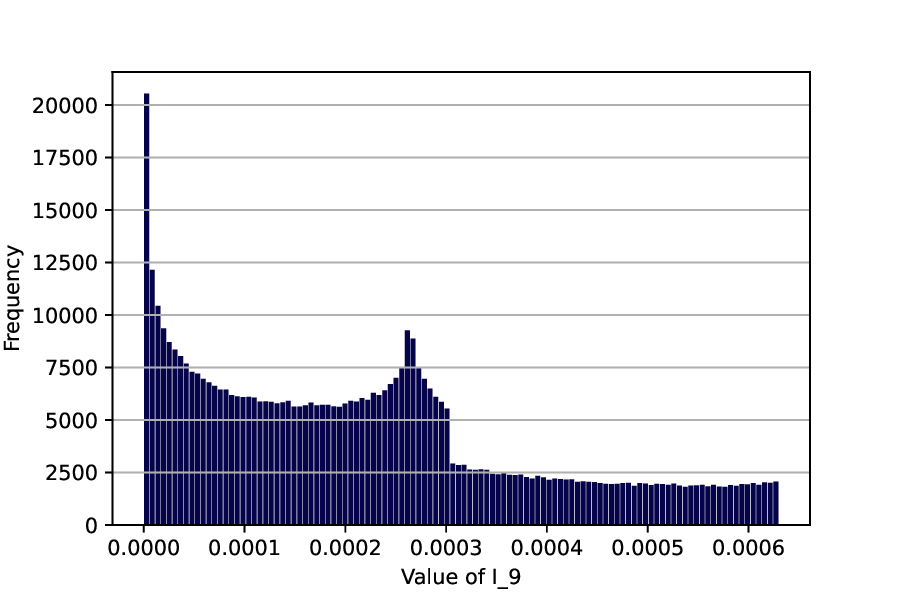}
         \caption{Histogram of the absolute value of fundamental invariant $I_{9}$ for 500,000 randomly generated generic real 3-qutrit states}
         \label{fig:i9_histo}
     \end{subfigure}
     \hfill
    \begin{subfigure}[b]{0.49\textwidth}
         \centering
         \includegraphics[width=\textwidth]{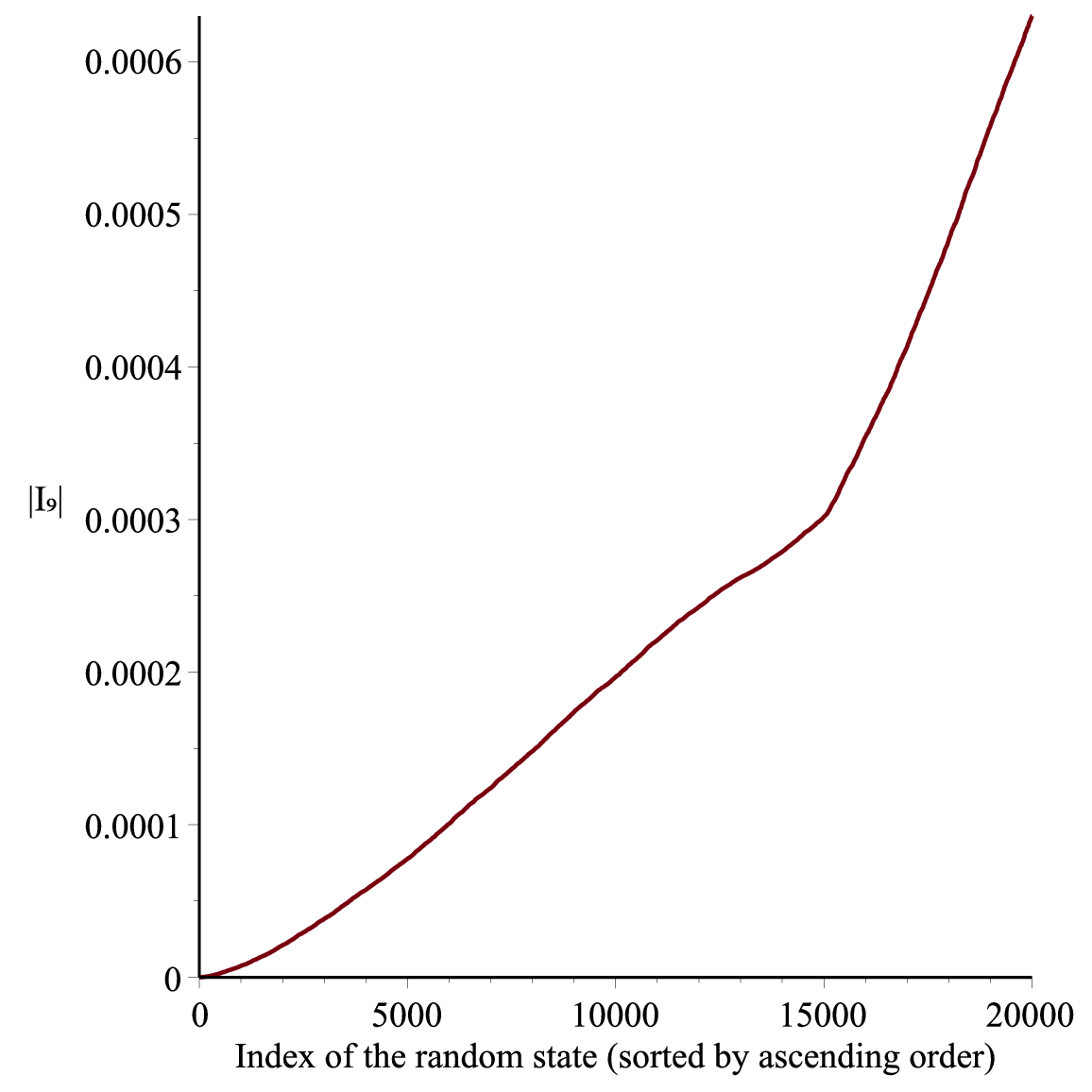}
         \caption{Absolute value of fundamental invariant $I_{9}$ for 20,000 randomly generated generic real 3-qutrit states, sorted in ascending order}
         \label{fig:i9_curve}
     \end{subfigure}
    \caption{}
    \label{fig:i9_plots}
\end{figure}

    \subsection{Random states and \texorpdfstring{$|I_{12}|$}{}}
The histogram \Cref{fig:i12_histo} and curve \Cref{fig:i12_curve} indicate that there are several regions where the invariant $I_{12}$ has similar values. In the sphere plot we also note the $a,b,c$ permutation symmetry. As in all the plots we notice that high values of the invariant are infrequent.
\begin{figure}[h]
    \centering
    \begin{subfigure}[b]{0.49\textwidth}
         \centering
         \includegraphics[width=\textwidth]{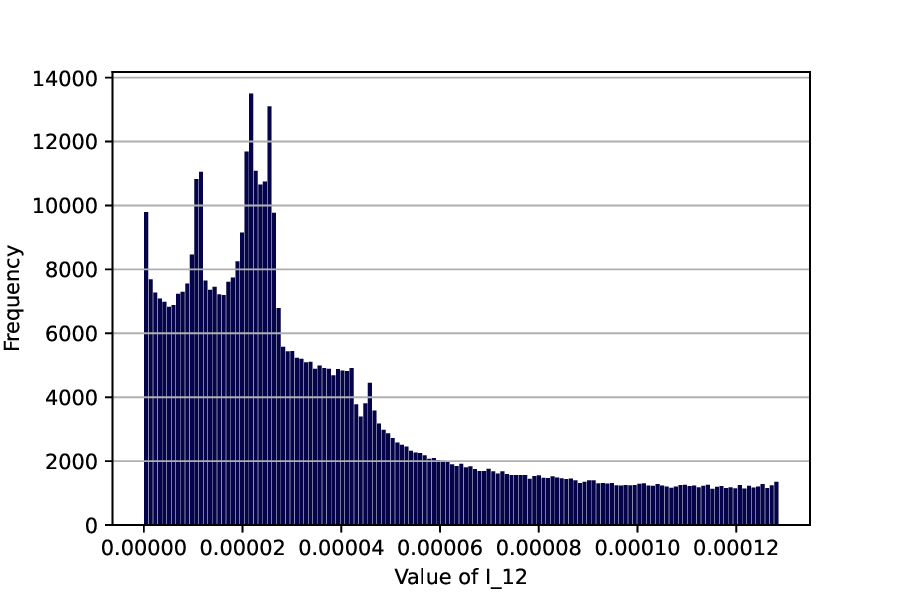}
         \caption{Histogram of the absolute value of fundamental invariant $I_{12}$ for 500,000 randomly generated generic real 3-qutrit states}
         \label{fig:i12_histo}
     \end{subfigure}
     \hfill
    \begin{subfigure}[b]{0.49\textwidth}
         \centering
         \includegraphics[width=\textwidth]{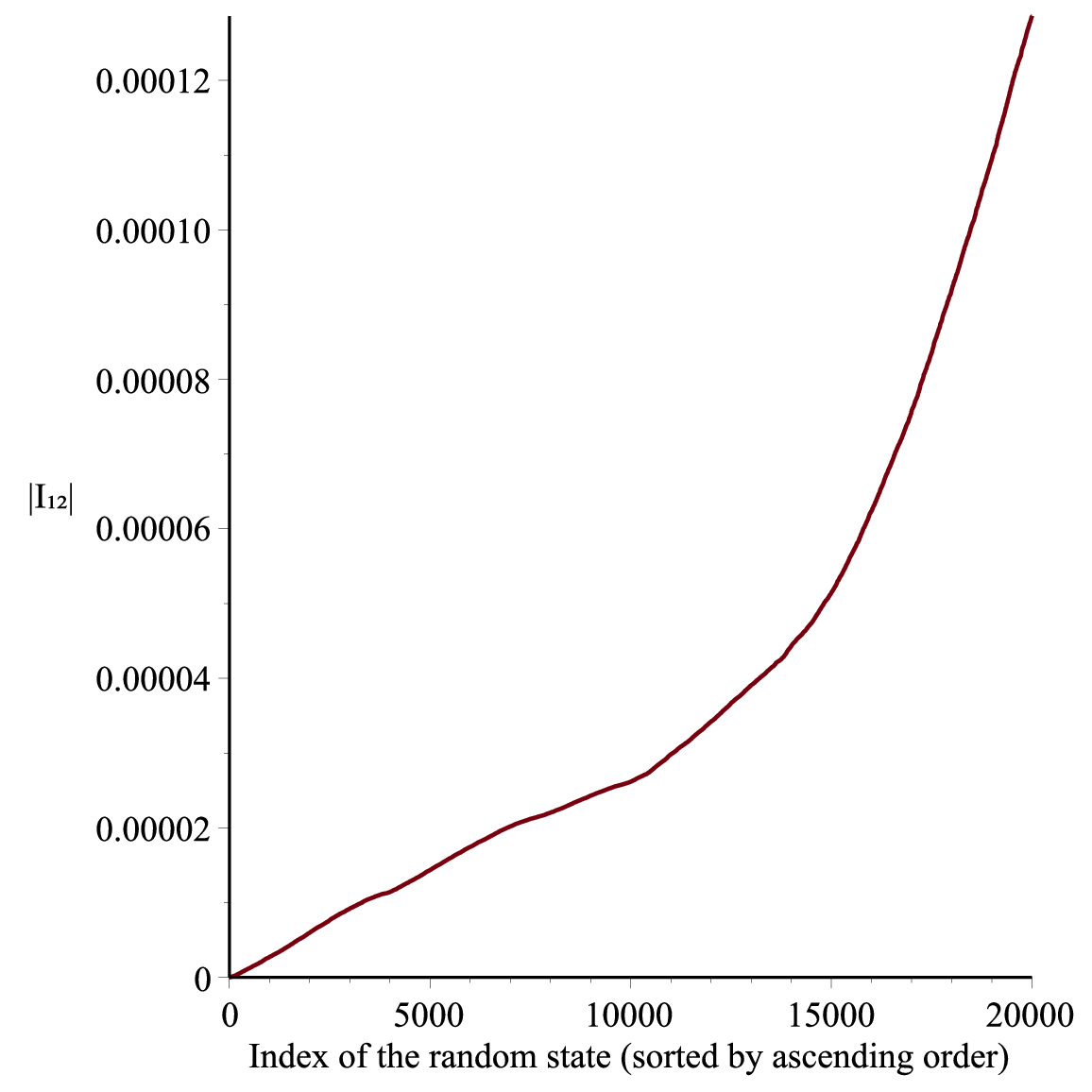}
         \caption{Absolute value of fundamental invariant $I_{12}$ for 20,000 randomly generated generic real 3-qutrit states, sorted in ascending order}
         \label{fig:i12_curve}
     \end{subfigure}
    \caption{}
    \label{fig:i12_plots}
\end{figure}

    \subsection{Combination of \texorpdfstring{$I_6$, $I_9$ \text{ and } $I_{12}$}{}}
Here we propose a way to combine the measurement of the fundamental invariants. 
We evaluate the quantity $S_I$, defined as the sum of absolute values of invariants divided by their respective maximum value
\begin{equation}
S_I = \frac{|I_6|}{m_{I_6}} +  \frac{|I_9|}{m_{I_9}} + \frac{|I_{12}|}{m_{I_{12}}}\;,\quad 
\text{with}\quad
    m_{I_6} = \frac{1}{18}, \; m_{I_9} = \frac{\sqrt{6}}{3888} , \; m_{I_{12}} = \frac{1}{7776} \;.
\end{equation}
For the Aharonov state $\ket{\mathcal A}$, which maximizes all three fundamental invariants, we have:
$S_I(\ket{ \mathcal{A}}) = 1 + 1 + 1 = 3$.
The histogram \Cref{fig:sum_histogram} shows that large values for this invariant are infrequent, with the most frequent value being approximately 0.7.

\begin{figure}[h]
    \centering
    \includegraphics[width=0.66\textwidth]{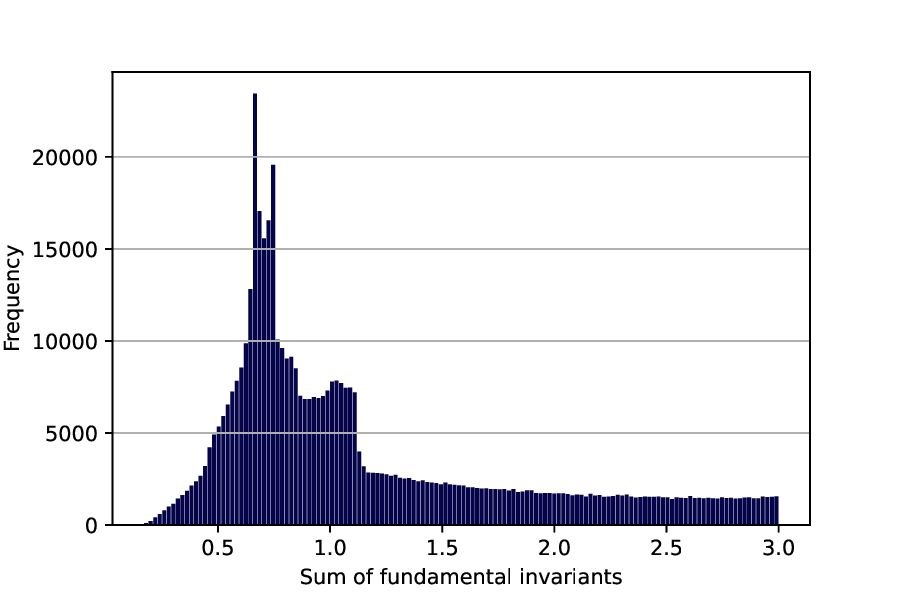}
    \caption{A histogram of the weighted sum of the absolute values of the invariants.}
    \label{fig:sum_histogram}
\end{figure}

\section{Evaluating Invariants on Any State}\label{sec:compare}
Until now in this article we have focused on evaluating invariants on the semi-simple part of a state. If one can compute the Jordan decomposition of a tensor, then the invariants have simple forms, and this is very useful for studying the properties of these invariants. On the other hand, if one is given a state and one does not know, or does not want to compute, its Jordan decomposition there is another way to evaluate the invariants at the given state. 
Here we explain how to use computations outlined in \cite{HolweckOeding2022} to evaluate the basic invariants via simple matrix operations (power-traces and determinants), and then use the formula in \cite{Bremner2014} to evaluate the hyperdeterminant. 

Tensors in $\CC^3 \otimes \CC^3 \otimes \CC^3$ can be viewed in the simple Lie algebra $\mathfrak{e}_6$ which has grading 
\[
\mathfrak{e}_6= \left(\sl_3^{\oplus 3}\right) \oplus \left(\CC^3 \otimes \CC^3 \otimes \CC^3\right) \oplus \left(\CC^3 \otimes \CC^3 \otimes \CC^3\right)^*\;.
\]
As such, we can construct the standard matrix representatives of the adjoint operators $\ad_Z(Y) = [Z,Y]$ using the bracket product from the Lie algebra. We have implemented these computations, which we described in more generality in \cite{HolweckOeding2022} in Macaulay2 \cite{M2} but recognizing that some readers may not be familiar with that system, we provide these matrices here and in a text file in the ancillary files to the arXiv version of this article. 

The matrix $K= \ad_Z(\cdot) = \begin{pmatrix} 0 & 0 & K_{0,2} \\ K_{1,0} & 0 & 0 \\ 0 & K_{2,1} & 0 \end{pmatrix}$ has blocks:

\noindent$K_{0,2}=$

\noindent\resizebox{\textwidth}{!}{$\left(\!\begin{array}{ccccccccccccccccccccccccccc}
       -\frac{1}{3}z_{2,2,2}&\frac{1}{3}z_{2,2,1}&-\frac{1}{3}z_{2,2,0}&\frac{1}{3}z_{2,1,2}&-\frac{1}{3}z_{2,1,1}&\frac{1}{3}z_{2,1,0}&-\frac{1}{3}z_{2,0,2}&\frac{1}{3}z_{2,0,1}&-\frac{1}{3}z_{2,0,0}&\frac{1}{3}z_{1,2,2}&-\frac{1}{3}z_{1,2,1}&\frac{1}{3}z_{1,2,0}&-\frac{1}{3}z_{1,1,2}&\frac{1}{3}z_{1,1,1}&-\frac{1}{3}z_{1,1,0}&\frac{1}{3}z_{1,0,2}&-\frac{1}{3}z_{1,0,1}&\frac{1}{3}z_{1,0,0}&\frac{2}{3}z_{0,2,2}&-\frac{2}{3}z_{0,2,1}&\frac{2}{3}z_{0,2,0}&-\frac{2}{3}z_{0,1,2}&\frac{2}{3}z_{0,1,1}&-\frac{2}{3}z_{0,1,0}&\frac{2}{3}z_{0,0,2}&-\frac{2}{3}z_{0,0,1}&\frac{2}{3}z_{0,0,0}\\
       -\frac{2}{3}z_{2,2,2}&\frac{2}{3}z_{2,2,1}&-\frac{2}{3}z_{2,2,0}&\frac{2}{3}z_{2,1,2}&-\frac{2}{3}z_{2,1,1}&\frac{2}{3}z_{2,1,0}&-\frac{2}{3}z_{2,0,2}&\frac{2}{3}z_{2,0,1}&-\frac{2}{3}z_{2,0,0}&-\frac{1}{3}z_{1,2,2}&\frac{1}{3}z_{1,2,1}&-\frac{1}{3}z_{1,2,0}&\frac{1}{3}z_{1,1,2}&-\frac{1}{3}z_{1,1,1}&\frac{1}{3}z_{1,1,0}&-\frac{1}{3}z_{1,0,2}&\frac{1}{3}z_{1,0,1}&-\frac{1}{3}z_{1,0,0}&\frac{1}{3}z_{0,2,2}&-\frac{1}{3}z_{0,2,1}&\frac{1}{3}z_{0,2,0}&-\frac{1}{3}z_{0,1,2}&\frac{1}{3}z_{0,1,1}&-\frac{1}{3}z_{0,1,0}&\frac{1}{3}z_{0,0,2}&-\frac{1}{3}z_{0,0,1}&\frac{1}{3}z_{0,0,0}\\
       -\frac{1}{3}z_{2,2,2}&\frac{1}{3}z_{2,2,1}&-\frac{1}{3}z_{2,2,0}&\frac{1}{3}z_{2,1,2}&-\frac{1}{3}z_{2,1,1}&\frac{1}{3}z_{2,1,0}&\frac{2}{3}z_{2,0,2}&-\frac{2}{3}z_{2,0,1}&\frac{2}{3}z_{2,0,0}&\frac{1}{3}z_{1,2,2}&-\frac{1}{3}z_{1,2,1}&\frac{1}{3}z_{1,2,0}&-\frac{1}{3}z_{1,1,2}&\frac{1}{3}z_{1,1,1}&-\frac{1}{3}z_{1,1,0}&-\frac{2}{3}z_{1,0,2}&\frac{2}{3}z_{1,0,1}&-\frac{2}{3}z_{1,0,0}&-\frac{1}{3}z_{0,2,2}&\frac{1}{3}z_{0,2,1}&-\frac{1}{3}z_{0,2,0}&\frac{1}{3}z_{0,1,2}&-\frac{1}{3}z_{0,1,1}&\frac{1}{3}z_{0,1,0}&\frac{2}{3}z_{0,0,2}&-\frac{2}{3}z_{0,0,1}&\frac{2}{3}z_{0,0,0}\\
       -\frac{2}{3}z_{2,2,2}&\frac{2}{3}z_{2,2,1}&-\frac{2}{3}z_{2,2,0}&-\frac{1}{3}z_{2,1,2}&\frac{1}{3}z_{2,1,1}&-\frac{1}{3}z_{2,1,0}&\frac{1}{3}z_{2,0,2}&-\frac{1}{3}z_{2,0,1}&\frac{1}{3}z_{2,0,0}&\frac{2}{3}z_{1,2,2}&-\frac{2}{3}z_{1,2,1}&\frac{2}{3}z_{1,2,0}&\frac{1}{3}z_{1,1,2}&-\frac{1}{3}z_{1,1,1}&\frac{1}{3}z_{1,1,0}&-\frac{1}{3}z_{1,0,2}&\frac{1}{3}z_{1,0,1}&-\frac{1}{3}z_{1,0,0}&-\frac{2}{3}z_{0,2,2}&\frac{2}{3}z_{0,2,1}&-\frac{2}{3}z_{0,2,0}&-\frac{1}{3}z_{0,1,2}&\frac{1}{3}z_{0,1,1}&-\frac{1}{3}z_{0,1,0}&\frac{1}{3}z_{0,0,2}&-\frac{1}{3}z_{0,0,1}&\frac{1}{3}z_{0,0,0}\\
       -\frac{1}{3}z_{2,2,2}&\frac{1}{3}z_{2,2,1}&\frac{2}{3}z_{2,2,0}&\frac{1}{3}z_{2,1,2}&-\frac{1}{3}z_{2,1,1}&-\frac{2}{3}z_{2,1,0}&-\frac{1}{3}z_{2,0,2}&\frac{1}{3}z_{2,0,1}&\frac{2}{3}z_{2,0,0}&\frac{1}{3}z_{1,2,2}&-\frac{1}{3}z_{1,2,1}&-\frac{2}{3}z_{1,2,0}&-\frac{1}{3}z_{1,1,2}&\frac{1}{3}z_{1,1,1}&\frac{2}{3}z_{1,1,0}&\frac{1}{3}z_{1,0,2}&-\frac{1}{3}z_{1,0,1}&-\frac{2}{3}z_{1,0,0}&-\frac{1}{3}z_{0,2,2}&\frac{1}{3}z_{0,2,1}&\frac{2}{3}z_{0,2,0}&\frac{1}{3}z_{0,1,2}&-\frac{1}{3}z_{0,1,1}&-\frac{2}{3}z_{0,1,0}&-\frac{1}{3}z_{0,0,2}&\frac{1}{3}z_{0,0,1}&\frac{2}{3}z_{0,0,0}\\
       -\frac{2}{3}z_{2,2,2}&-\frac{1}{3}z_{2,2,1}&\frac{1}{3}z_{2,2,0}&\frac{2}{3}z_{2,1,2}&\frac{1}{3}z_{2,1,1}&-\frac{1}{3}z_{2,1,0}&-\frac{2}{3}z_{2,0,2}&-\frac{1}{3}z_{2,0,1}&\frac{1}{3}z_{2,0,0}&\frac{2}{3}z_{1,2,2}&\frac{1}{3}z_{1,2,1}&-\frac{1}{3}z_{1,2,0}&-\frac{2}{3}z_{1,1,2}&-\frac{1}{3}z_{1,1,1}&\frac{1}{3}z_{1,1,0}&\frac{2}{3}z_{1,0,2}&\frac{1}{3}z_{1,0,1}&-\frac{1}{3}z_{1,0,0}&-\frac{2}{3}z_{0,2,2}&-\frac{1}{3}z_{0,2,1}&\frac{1}{3}z_{0,2,0}&\frac{2}{3}z_{0,1,2}&\frac{1}{3}z_{0,1,1}&-\frac{1}{3}z_{0,1,0}&-\frac{2}{3}z_{0,0,2}&-\frac{1}{3}z_{0,0,1}&\frac{1}{3}z_{0,0,0}\\
       0&0&0&0&0&0&0&0&0&-z_{0,2,2}&z_{0,2,1}&-z_{0,2,0}&z_{0,1,2}&-z_{0,1,1}&z_{0,1,0}&-z_{0,0,2}&z_{0,0,1}&-z_{0,0,0}&0&0&0&0&0&0&0&0&0\\
       z_{0,2,2}&-z_{0,2,1}&z_{0,2,0}&-z_{0,1,2}&z_{0,1,1}&-z_{0,1,0}&z_{0,0,2}&-z_{0,0,1}&z_{0,0,0}&0&0&0&0&0&0&0&0&0&0&0&0&0&0&0&0&0&0\\
       z_{1,2,2}&-z_{1,2,1}&z_{1,2,0}&-z_{1,1,2}&z_{1,1,1}&-z_{1,1,0}&z_{1,0,2}&-z_{1,0,1}&z_{1,0,0}&0&0&0&0&0&0&0&0&0&0&0&0&0&0&0&0&0&0\\
       0&0&0&-z_{2,0,2}&z_{2,0,1}&-z_{2,0,0}&0&0&0&0&0&0&z_{1,0,2}&-z_{1,0,1}&z_{1,0,0}&0&0&0&0&0&0&-z_{0,0,2}&z_{0,0,1}&-z_{0,0,0}&0&0&0\\
       z_{2,0,2}&-z_{2,0,1}&z_{2,0,0}&0&0&0&0&0&0&-z_{1,0,2}&z_{1,0,1}&-z_{1,0,0}&0&0&0&0&0&0&z_{0,0,2}&-z_{0,0,1}&z_{0,0,0}&0&0&0&0&0&0\\
       z_{2,1,2}&-z_{2,1,1}&z_{2,1,0}&0&0&0&0&0&0&-z_{1,1,2}&z_{1,1,1}&-z_{1,1,0}&0&0&0&0&0&0&z_{0,1,2}&-z_{0,1,1}&z_{0,1,0}&0&0&0&0&0&0\\
       0&-z_{2,2,0}&0&0&z_{2,1,0}&0&0&-z_{2,0,0}&0&0&z_{1,2,0}&0&0&-z_{1,1,0}&0&0&z_{1,0,0}&0&0&-z_{0,2,0}&0&0&z_{0,1,0}&0&0&-z_{0,0,0}&0\\
       z_{2,2,0}&0&0&-z_{2,1,0}&0&0&z_{2,0,0}&0&0&-z_{1,2,0}&0&0&z_{1,1,0}&0&0&-z_{1,0,0}&0&0&z_{0,2,0}&0&0&-z_{0,1,0}&0&0&z_{0,0,0}&0&0\\
       z_{2,2,1}&0&0&-z_{2,1,1}&0&0&z_{2,0,1}&0&0&-z_{1,2,1}&0&0&z_{1,1,1}&0&0&-z_{1,0,1}&0&0&z_{0,2,1}&0&0&-z_{0,1,1}&0&0&z_{0,0,1}&0&0\\
       0&0&0&0&0&0&0&0&0&0&0&0&0&0&0&0&0&0&z_{1,2,2}&-z_{1,2,1}&z_{1,2,0}&-z_{1,1,2}&z_{1,1,1}&-z_{1,1,0}&z_{1,0,2}&-z_{1,0,1}&z_{1,0,0}\\
       0&0&0&0&0&0&0&0&0&0&0&0&0&0&0&0&0&0&z_{2,2,2}&-z_{2,2,1}&z_{2,2,0}&-z_{2,1,2}&z_{2,1,1}&-z_{2,1,0}&z_{2,0,2}&-z_{2,0,1}&z_{2,0,0}\\
       0&0&0&0&0&0&0&0&0&-z_{2,2,2}&z_{2,2,1}&-z_{2,2,0}&z_{2,1,2}&-z_{2,1,1}&z_{2,1,0}&-z_{2,0,2}&z_{2,0,1}&-z_{2,0,0}&0&0&0&0&0&0&0&0&0\\
       0&0&0&0&0&0&z_{2,1,2}&-z_{2,1,1}&z_{2,1,0}&0&0&0&0&0&0&-z_{1,1,2}&z_{1,1,1}&-z_{1,1,0}&0&0&0&0&0&0&z_{0,1,2}&-z_{0,1,1}&z_{0,1,0}\\
       0&0&0&0&0&0&z_{2,2,2}&-z_{2,2,1}&z_{2,2,0}&0&0&0&0&0&0&-z_{1,2,2}&z_{1,2,1}&-z_{1,2,0}&0&0&0&0&0&0&z_{0,2,2}&-z_{0,2,1}&z_{0,2,0}\\
       0&0&0&-z_{2,2,2}&z_{2,2,1}&-z_{2,2,0}&0&0&0&0&0&0&z_{1,2,2}&-z_{1,2,1}&z_{1,2,0}&0&0&0&0&0&0&-z_{0,2,2}&z_{0,2,1}&-z_{0,2,0}&0&0&0\\
       0&0&z_{2,2,1}&0&0&-z_{2,1,1}&0&0&z_{2,0,1}&0&0&-z_{1,2,1}&0&0&z_{1,1,1}&0&0&-z_{1,0,1}&0&0&z_{0,2,1}&0&0&-z_{0,1,1}&0&0&z_{0,0,1}\\
       0&0&z_{2,2,2}&0&0&-z_{2,1,2}&0&0&z_{2,0,2}&0&0&-z_{1,2,2}&0&0&z_{1,1,2}&0&0&-z_{1,0,2}&0&0&z_{0,2,2}&0&0&-z_{0,1,2}&0&0&z_{0,0,2}\\
       0&-z_{2,2,2}&0&0&z_{2,1,2}&0&0&-z_{2,0,2}&0&0&z_{1,2,2}&0&0&-z_{1,1,2}&0&0&z_{1,0,2}&0&0&-z_{0,2,2}&0&0&z_{0,1,2}&0&0&-z_{0,0,2}&0
       \end{array}\!\right)$
}

~

\noindent$K_{1,0}=$

\noindent\resizebox{\textwidth}{!}{
$\left(\!\begin{array}{cccccccccccccccccccccccc}
       -z_{0,0,0}&0&-z_{0,0,0}&0&-z_{0,0,0}&0&-z_{1,0,0}&-z_{2,0,0}&0&-z_{0,1,0}&-z_{0,2,0}&0&-z_{0,0,1}&-z_{0,0,2}&0&0&0&0&0&0&0&0&0&0\\
       -z_{0,0,1}&0&-z_{0,0,1}&0&z_{0,0,1}&-z_{0,0,1}&-z_{1,0,1}&-z_{2,0,1}&0&-z_{0,1,1}&-z_{0,2,1}&0&0&0&-z_{0,0,2}&0&0&0&0&0&0&-z_{0,0,0}&0&0\\
       -z_{0,0,2}&0&-z_{0,0,2}&0&0&z_{0,0,2}&-z_{1,0,2}&-z_{2,0,2}&0&-z_{0,1,2}&-z_{0,2,2}&0&0&0&0&0&0&0&0&0&0&0&-z_{0,0,0}&-z_{0,0,1}\\
       -z_{0,1,0}&0&z_{0,1,0}&-z_{0,1,0}&-z_{0,1,0}&0&-z_{1,1,0}&-z_{2,1,0}&0&0&0&-z_{0,2,0}&-z_{0,1,1}&-z_{0,1,2}&0&0&0&0&-z_{0,0,0}&0&0&0&0&0\\
       -z_{0,1,1}&0&z_{0,1,1}&-z_{0,1,1}&z_{0,1,1}&-z_{0,1,1}&-z_{1,1,1}&-z_{2,1,1}&0&0&0&-z_{0,2,1}&0&0&-z_{0,1,2}&0&0&0&-z_{0,0,1}&0&0&-z_{0,1,0}&0&0\\
       -z_{0,1,2}&0&z_{0,1,2}&-z_{0,1,2}&0&z_{0,1,2}&-z_{1,1,2}&-z_{2,1,2}&0&0&0&-z_{0,2,2}&0&0&0&0&0&0&-z_{0,0,2}&0&0&0&-z_{0,1,0}&-z_{0,1,1}\\
       -z_{0,2,0}&0&0&z_{0,2,0}&-z_{0,2,0}&0&-z_{1,2,0}&-z_{2,2,0}&0&0&0&0&-z_{0,2,1}&-z_{0,2,2}&0&0&0&0&0&-z_{0,0,0}&-z_{0,1,0}&0&0&0\\
       -z_{0,2,1}&0&0&z_{0,2,1}&z_{0,2,1}&-z_{0,2,1}&-z_{1,2,1}&-z_{2,2,1}&0&0&0&0&0&0&-z_{0,2,2}&0&0&0&0&-z_{0,0,1}&-z_{0,1,1}&-z_{0,2,0}&0&0\\
       -z_{0,2,2}&0&0&z_{0,2,2}&0&z_{0,2,2}&-z_{1,2,2}&-z_{2,2,2}&0&0&0&0&0&0&0&0&0&0&0&-z_{0,0,2}&-z_{0,1,2}&0&-z_{0,2,0}&-z_{0,2,1}\\
       z_{1,0,0}&-z_{1,0,0}&-z_{1,0,0}&0&-z_{1,0,0}&0&0&0&-z_{2,0,0}&-z_{1,1,0}&-z_{1,2,0}&0&-z_{1,0,1}&-z_{1,0,2}&0&-z_{0,0,0}&0&0&0&0&0&0&0&0\\
       z_{1,0,1}&-z_{1,0,1}&-z_{1,0,1}&0&z_{1,0,1}&-z_{1,0,1}&0&0&-z_{2,0,1}&-z_{1,1,1}&-z_{1,2,1}&0&0&0&-z_{1,0,2}&-z_{0,0,1}&0&0&0&0&0&-z_{1,0,0}&0&0\\
       z_{1,0,2}&-z_{1,0,2}&-z_{1,0,2}&0&0&z_{1,0,2}&0&0&-z_{2,0,2}&-z_{1,1,2}&-z_{1,2,2}&0&0&0&0&-z_{0,0,2}&0&0&0&0&0&0&-z_{1,0,0}&-z_{1,0,1}\\
       z_{1,1,0}&-z_{1,1,0}&z_{1,1,0}&-z_{1,1,0}&-z_{1,1,0}&0&0&0&-z_{2,1,0}&0&0&-z_{1,2,0}&-z_{1,1,1}&-z_{1,1,2}&0&-z_{0,1,0}&0&0&-z_{1,0,0}&0&0&0&0&0\\
       z_{1,1,1}&-z_{1,1,1}&z_{1,1,1}&-z_{1,1,1}&z_{1,1,1}&-z_{1,1,1}&0&0&-z_{2,1,1}&0&0&-z_{1,2,1}&0&0&-z_{1,1,2}&-z_{0,1,1}&0&0&-z_{1,0,1}&0&0&-z_{1,1,0}&0&0\\
       z_{1,1,2}&-z_{1,1,2}&z_{1,1,2}&-z_{1,1,2}&0&z_{1,1,2}&0&0&-z_{2,1,2}&0&0&-z_{1,2,2}&0&0&0&-z_{0,1,2}&0&0&-z_{1,0,2}&0&0&0&-z_{1,1,0}&-z_{1,1,1}\\
       z_{1,2,0}&-z_{1,2,0}&0&z_{1,2,0}&-z_{1,2,0}&0&0&0&-z_{2,2,0}&0&0&0&-z_{1,2,1}&-z_{1,2,2}&0&-z_{0,2,0}&0&0&0&-z_{1,0,0}&-z_{1,1,0}&0&0&0\\
       z_{1,2,1}&-z_{1,2,1}&0&z_{1,2,1}&z_{1,2,1}&-z_{1,2,1}&0&0&-z_{2,2,1}&0&0&0&0&0&-z_{1,2,2}&-z_{0,2,1}&0&0&0&-z_{1,0,1}&-z_{1,1,1}&-z_{1,2,0}&0&0\\
       z_{1,2,2}&-z_{1,2,2}&0&z_{1,2,2}&0&z_{1,2,2}&0&0&-z_{2,2,2}&0&0&0&0&0&0&-z_{0,2,2}&0&0&0&-z_{1,0,2}&-z_{1,1,2}&0&-z_{1,2,0}&-z_{1,2,1}\\
       0&z_{2,0,0}&-z_{2,0,0}&0&-z_{2,0,0}&0&0&0&0&-z_{2,1,0}&-z_{2,2,0}&0&-z_{2,0,1}&-z_{2,0,2}&0&0&-z_{0,0,0}&-z_{1,0,0}&0&0&0&0&0&0\\
       0&z_{2,0,1}&-z_{2,0,1}&0&z_{2,0,1}&-z_{2,0,1}&0&0&0&-z_{2,1,1}&-z_{2,2,1}&0&0&0&-z_{2,0,2}&0&-z_{0,0,1}&-z_{1,0,1}&0&0&0&-z_{2,0,0}&0&0\\
       0&z_{2,0,2}&-z_{2,0,2}&0&0&z_{2,0,2}&0&0&0&-z_{2,1,2}&-z_{2,2,2}&0&0&0&0&0&-z_{0,0,2}&-z_{1,0,2}&0&0&0&0&-z_{2,0,0}&-z_{2,0,1}\\
       0&z_{2,1,0}&z_{2,1,0}&-z_{2,1,0}&-z_{2,1,0}&0&0&0&0&0&0&-z_{2,2,0}&-z_{2,1,1}&-z_{2,1,2}&0&0&-z_{0,1,0}&-z_{1,1,0}&-z_{2,0,0}&0&0&0&0&0\\
       0&z_{2,1,1}&z_{2,1,1}&-z_{2,1,1}&z_{2,1,1}&-z_{2,1,1}&0&0&0&0&0&-z_{2,2,1}&0&0&-z_{2,1,2}&0&-z_{0,1,1}&-z_{1,1,1}&-z_{2,0,1}&0&0&-z_{2,1,0}&0&0\\
       0&z_{2,1,2}&z_{2,1,2}&-z_{2,1,2}&0&z_{2,1,2}&0&0&0&0&0&-z_{2,2,2}&0&0&0&0&-z_{0,1,2}&-z_{1,1,2}&-z_{2,0,2}&0&0&0&-z_{2,1,0}&-z_{2,1,1}\\
       0&z_{2,2,0}&0&z_{2,2,0}&-z_{2,2,0}&0&0&0&0&0&0&0&-z_{2,2,1}&-z_{2,2,2}&0&0&-z_{0,2,0}&-z_{1,2,0}&0&-z_{2,0,0}&-z_{2,1,0}&0&0&0\\
       0&z_{2,2,1}&0&z_{2,2,1}&z_{2,2,1}&-z_{2,2,1}&0&0&0&0&0&0&0&0&-z_{2,2,2}&0&-z_{0,2,1}&-z_{1,2,1}&0&-z_{2,0,1}&-z_{2,1,1}&-z_{2,2,0}&0&0\\
       0&z_{2,2,2}&0&z_{2,2,2}&0&z_{2,2,2}&0&0&0&0&0&0&0&0&0&0&-z_{0,2,2}&-z_{1,2,2}&0&-z_{2,0,2}&-z_{2,1,2}&0&-z_{2,2,0}&-z_{2,2,1}
       \end{array}\!\right)$}

~

\noindent$K_{2,1}=$

\noindent\resizebox{\textwidth}{!}{$\left(\!\begin{array}{ccccccccccccccccccccccccccc}
       z_{1,1,1}&-z_{1,1,0}&0&-z_{1,0,1}&z_{1,0,0}&0&0&0&0&-z_{0,1,1}&z_{0,1,0}&0&z_{0,0,1}&-z_{0,0,0}&0&0&0&0&0&0&0&0&0&0&0&0&0\\
       z_{1,1,2}&0&-z_{1,1,0}&-z_{1,0,2}&0&z_{1,0,0}&0&0&0&-z_{0,1,2}&0&z_{0,1,0}&z_{0,0,2}&0&-z_{0,0,0}&0&0&0&0&0&0&0&0&0&0&0&0\\
       0&z_{1,1,2}&-z_{1,1,1}&0&-z_{1,0,2}&z_{1,0,1}&0&0&0&0&-z_{0,1,2}&z_{0,1,1}&0&z_{0,0,2}&-z_{0,0,1}&0&0&0&0&0&0&0&0&0&0&0&0\\
       z_{1,2,1}&-z_{1,2,0}&0&0&0&0&-z_{1,0,1}&z_{1,0,0}&0&-z_{0,2,1}&z_{0,2,0}&0&0&0&0&z_{0,0,1}&-z_{0,0,0}&0&0&0&0&0&0&0&0&0&0\\
       z_{1,2,2}&0&-z_{1,2,0}&0&0&0&-z_{1,0,2}&0&z_{1,0,0}&-z_{0,2,2}&0&z_{0,2,0}&0&0&0&z_{0,0,2}&0&-z_{0,0,0}&0&0&0&0&0&0&0&0&0\\
       0&z_{1,2,2}&-z_{1,2,1}&0&0&0&0&-z_{1,0,2}&z_{1,0,1}&0&-z_{0,2,2}&z_{0,2,1}&0&0&0&0&z_{0,0,2}&-z_{0,0,1}&0&0&0&0&0&0&0&0&0\\
       0&0&0&z_{1,2,1}&-z_{1,2,0}&0&-z_{1,1,1}&z_{1,1,0}&0&0&0&0&-z_{0,2,1}&z_{0,2,0}&0&z_{0,1,1}&-z_{0,1,0}&0&0&0&0&0&0&0&0&0&0\\
       0&0&0&z_{1,2,2}&0&-z_{1,2,0}&-z_{1,1,2}&0&z_{1,1,0}&0&0&0&-z_{0,2,2}&0&z_{0,2,0}&z_{0,1,2}&0&-z_{0,1,0}&0&0&0&0&0&0&0&0&0\\
       0&0&0&0&z_{1,2,2}&-z_{1,2,1}&0&-z_{1,1,2}&z_{1,1,1}&0&0&0&0&-z_{0,2,2}&z_{0,2,1}&0&z_{0,1,2}&-z_{0,1,1}&0&0&0&0&0&0&0&0&0\\
       z_{2,1,1}&-z_{2,1,0}&0&-z_{2,0,1}&z_{2,0,0}&0&0&0&0&0&0&0&0&0&0&0&0&0&-z_{0,1,1}&z_{0,1,0}&0&z_{0,0,1}&-z_{0,0,0}&0&0&0&0\\
       z_{2,1,2}&0&-z_{2,1,0}&-z_{2,0,2}&0&z_{2,0,0}&0&0&0&0&0&0&0&0&0&0&0&0&-z_{0,1,2}&0&z_{0,1,0}&z_{0,0,2}&0&-z_{0,0,0}&0&0&0\\
       0&z_{2,1,2}&-z_{2,1,1}&0&-z_{2,0,2}&z_{2,0,1}&0&0&0&0&0&0&0&0&0&0&0&0&0&-z_{0,1,2}&z_{0,1,1}&0&z_{0,0,2}&-z_{0,0,1}&0&0&0\\
       z_{2,2,1}&-z_{2,2,0}&0&0&0&0&-z_{2,0,1}&z_{2,0,0}&0&0&0&0&0&0&0&0&0&0&-z_{0,2,1}&z_{0,2,0}&0&0&0&0&z_{0,0,1}&-z_{0,0,0}&0\\
       z_{2,2,2}&0&-z_{2,2,0}&0&0&0&-z_{2,0,2}&0&z_{2,0,0}&0&0&0&0&0&0&0&0&0&-z_{0,2,2}&0&z_{0,2,0}&0&0&0&z_{0,0,2}&0&-z_{0,0,0}\\
       0&z_{2,2,2}&-z_{2,2,1}&0&0&0&0&-z_{2,0,2}&z_{2,0,1}&0&0&0&0&0&0&0&0&0&0&-z_{0,2,2}&z_{0,2,1}&0&0&0&0&z_{0,0,2}&-z_{0,0,1}\\
       0&0&0&z_{2,2,1}&-z_{2,2,0}&0&-z_{2,1,1}&z_{2,1,0}&0&0&0&0&0&0&0&0&0&0&0&0&0&-z_{0,2,1}&z_{0,2,0}&0&z_{0,1,1}&-z_{0,1,0}&0\\
       0&0&0&z_{2,2,2}&0&-z_{2,2,0}&-z_{2,1,2}&0&z_{2,1,0}&0&0&0&0&0&0&0&0&0&0&0&0&-z_{0,2,2}&0&z_{0,2,0}&z_{0,1,2}&0&-z_{0,1,0}\\
       0&0&0&0&z_{2,2,2}&-z_{2,2,1}&0&-z_{2,1,2}&z_{2,1,1}&0&0&0&0&0&0&0&0&0&0&0&0&0&-z_{0,2,2}&z_{0,2,1}&0&z_{0,1,2}&-z_{0,1,1}\\
       0&0&0&0&0&0&0&0&0&z_{2,1,1}&-z_{2,1,0}&0&-z_{2,0,1}&z_{2,0,0}&0&0&0&0&-z_{1,1,1}&z_{1,1,0}&0&z_{1,0,1}&-z_{1,0,0}&0&0&0&0\\
       0&0&0&0&0&0&0&0&0&z_{2,1,2}&0&-z_{2,1,0}&-z_{2,0,2}&0&z_{2,0,0}&0&0&0&-z_{1,1,2}&0&z_{1,1,0}&z_{1,0,2}&0&-z_{1,0,0}&0&0&0\\
       0&0&0&0&0&0&0&0&0&0&z_{2,1,2}&-z_{2,1,1}&0&-z_{2,0,2}&z_{2,0,1}&0&0&0&0&-z_{1,1,2}&z_{1,1,1}&0&z_{1,0,2}&-z_{1,0,1}&0&0&0\\
       0&0&0&0&0&0&0&0&0&z_{2,2,1}&-z_{2,2,0}&0&0&0&0&-z_{2,0,1}&z_{2,0,0}&0&-z_{1,2,1}&z_{1,2,0}&0&0&0&0&z_{1,0,1}&-z_{1,0,0}&0\\
       0&0&0&0&0&0&0&0&0&z_{2,2,2}&0&-z_{2,2,0}&0&0&0&-z_{2,0,2}&0&z_{2,0,0}&-z_{1,2,2}&0&z_{1,2,0}&0&0&0&z_{1,0,2}&0&-z_{1,0,0}\\
       0&0&0&0&0&0&0&0&0&0&z_{2,2,2}&-z_{2,2,1}&0&0&0&0&-z_{2,0,2}&z_{2,0,1}&0&-z_{1,2,2}&z_{1,2,1}&0&0&0&0&z_{1,0,2}&-z_{1,0,1}\\
       0&0&0&0&0&0&0&0&0&0&0&0&z_{2,2,1}&-z_{2,2,0}&0&-z_{2,1,1}&z_{2,1,0}&0&0&0&0&-z_{1,2,1}&z_{1,2,0}&0&z_{1,1,1}&-z_{1,1,0}&0\\
       0&0&0&0&0&0&0&0&0&0&0&0&z_{2,2,2}&0&-z_{2,2,0}&-z_{2,1,2}&0&z_{2,1,0}&0&0&0&-z_{1,2,2}&0&z_{1,2,0}&z_{1,1,2}&0&-z_{1,1,0}\\
       0&0&0&0&0&0&0&0&0&0&0&0&0&z_{2,2,2}&-z_{2,2,1}&0&-z_{2,1,2}&z_{2,1,1}&0&0&0&0&-z_{1,2,2}&z_{1,2,1}&0&z_{1,1,2}&-z_{1,1,1}
       \end{array}\!\right)$}

Note that by construction the trace of powers of $K$ are invariants, and we find that $g_6:=\tr(K^6)$ and $g_{12} :=\tr(K^{12})$ are non-zero, and algebraically independent. 
Unfortunately, in degree 9 we find that $\tr{K^9}=0$, so we must obtain the degree 9 invariant another way. 
To do this we compute $g_9 := \det(S_9)$, where $S_9$ is the Strassen matrix:
\[S_9 = 
\left(\!\begin{array}{ccccccccc}
       0&0&0&-z_{2,0,0}&-z_{2,0,1}&-z_{2,0,2}&z_{1,0,0}&z_{1,0,1}&z_{1,0,2}\\
       0&0&0&-z_{2,1,0}&-z_{2,1,1}&-z_{2,1,2}&z_{1,1,0}&z_{1,1,1}&z_{1,1,2}\\
       0&0&0&-z_{2,2,0}&-z_{2,2,1}&-z_{2,2,2}&z_{1,2,0}&z_{1,2,1}&z_{1,2,2}\\
       z_{2,0,0}&z_{2,0,1}&z_{2,0,2}&0&0&0&-z_{0,0,0}&-z_{0,0,1}&-z_{0,0,2}\\
       z_{2,1,0}&z_{2,1,1}&z_{2,1,2}&0&0&0&-z_{0,1,0}&-z_{0,1,1}&-z_{0,1,2}\\
       z_{2,2,0}&z_{2,2,1}&z_{2,2,2}&0&0&0&-z_{0,2,0}&-z_{0,2,1}&-z_{0,2,2}\\
       -z_{1,0,0}&-z_{1,0,1}&-z_{1,0,2}&z_{0,0,0}&z_{0,0,1}&z_{0,0,2}&0&0&0\\
       -z_{1,1,0}&-z_{1,1,1}&-z_{1,1,2}&z_{0,1,0}&z_{0,1,1}&z_{0,1,2}&0&0&0\\
       -z_{1,2,0}&-z_{1,2,1}&-z_{1,2,2}&z_{0,2,0}&z_{0,2,1}&z_{0,2,2}&0&0&0
       \end{array}\!\right)\;.
\]
To match the conventions as in \cite{Bremner2014} we apply the following changes:

\[
I_6 := \frac{g_6}{-108}, 
\quad I_9 := -g_9, \quad
I_{12} :=  \frac{1}{930}\cdot \left(\frac{g_{12}}{108}- 41\cdot \left(\frac{g_6}{-108}\right)^2\right)\;.
\]
Then we can use the formula in \cite{Bremner2014}*{Theorem~3.1} for the $3\times 3\times 3$ hyperdeterminant:
\[\Delta_{333}=  I_6^3I_9^2 - I_6^2I_{12}^2 + 36I_6I_9^2I_{12} + 108I_9^4 - 32I_{12}^3\;.
\]

Evaluating the invariants on the forms at \Cref{eq:psi1,eq:psi2,eq:D3,eq:ghz,eq:A,eq:psi3,eq:d2d3} we find the following values.
\[\renewcommand{\arraystretch}{1.5}\begin{array}{c||c|c|c|c|c|c|c|c||}
    \text{}& \ket{\Psi_1} & \ket{\Psi_2} & \ket{D[3,(1,1,1)]} & \ket{\text{GHZ}_{333}} & \ket{\mathcal{A}} & \ket{\psi_3} & \ket{D_3^2} & \ket{D_3^3}
    \\ \hline
|I_6| & 0.007245 
 & 0.009383 
& {\frac{1}{27}} 
& \frac{1}{27}
& \frac{1}{18}
& 0
& 0
& \frac{1}{125}

\\
    |I_9| & 7.954\cdot 10^{-5} 
& 5.712\cdot 10^{-5} 
& 0 
& 0
& \frac{\sqrt{6}}{3888}
&  0
& 0
& 0
\\
|I_{12}| &
5.245\cdot 10^{-6} & 4.332\cdot 10^{-6} & {\frac{1}{23328}} = \frac{1}{2^5\cdot 3^6} & 0 & \frac{1}{7776} 
& 0
& 0
& \frac{1}{500000}
\\
|\Delta_{333}| &
6.243\cdot 10^{-16} & 7.889\cdot 10^{-17} & 0 & 0 & 0 & 0  & 0 & 0
\end{array}
\]
Note that both $\ket{W}$ and $\ket{W_{333}}$ at \eqref{eq:W} states are nilpotent and hence annihilate all continuous invariants. 
We repeat what was noted in \Cref{rem:critPts} that the Aharonov state $\ket{\mathcal A}$ \eqref{eq:A} maximizes the 3 fundamental invariants, and the states $\ket{D[3,(1,1,1)]}$ \eqref{eq:A} and $\ket{\text{GHZ}_{333}}$ \eqref{eq:ghz} are critical points for $|I_6|$. 

One may check these values with the code we provide in the ancillary files of the arXiv version of this article.  By numerically perturbing each we also checked that none of $\ket{\Psi_1},$ $ \ket{\Psi_2},$ $\ket{\text{GHZ}_{333}},$ $  \ket{\psi_3}$, $\ket{D_3^2},$ $ \ket{D_3^3}$ are critical points of these invariants, however $\ket {D[3,(1,1,1)]}$ is a critical point of $\Delta_{333}$ and  $\ket{\mathcal{A}}$ is a critical point of the fundamental invariants.

For example, starting from $\ket{\Psi_1}$ we obtain the state $\ket{\tilde \Psi_1}$ defined below, which has value $|\Delta_{333}(\ket{\tilde \Psi_1})| = 6.9069\times 10^{-13}$ that is a near max for $\Delta_{333}$. $\ket{\tilde \Psi_1}=$

\noindent\resizebox{\textwidth}{!}{$\begin{array}{rrr}
( 0.039366- 0.023753 i ) \ket{000}&
+(- 0.111693 - 0.197348 i )\ket{001}&
+(- 0.122949- 0.208861 i) \ket{002} \\ 
+(- 0.095285- 0.163765 i) \ket{010}&
+(- 0.174350- 0.009913 i ) \ket{011}&
+(- 0.105943 - 0.112670 i)\ket{012}  \\ 
+(- 0.209619 - 0.105469 i) \ket{020}&
+ (0.199823- 0.155145 i ) \ket{021}&
+( 0.101870 - 0.036988 i )\ket{022} \\ 
+(- 0.077520 + 0.002270 i )\ket{100}&
+(- 0.074101 - 0.083099 i ) \ket{101}&
+(- 0.008464 + 0.125058 i)\ket{102}  \\ 
+(- 0.024819 - 0.278482 i ) \ket{110}&
+( 0.188483 + 0.174583 i)\ket{111}&
(- 0.136120 - 0.188541 i) \ket{112} \\ 
+( 0.043834- 0.109649 i) \ket{120}&
+(- 0.095310 + 0.006777 i)\ket{121}&  
+(- 0.205072 + 0.188972 i)\ket{122} \\ 
+(- 0.327100 - 0.115437 i)\ket{200}&
+(- 0.130088 + 0.126213 i ) \ket{201}&
+ (0.093746 - 0.158298 i) \ket{202} \\ 
+ (0.012939 - 0.056325 i) \ket{210}&
+(- 0.092719 + 0.003978 i )\ket{211}&
+(- 0.175832 + 0.043777 i)\ket{212} \\ 
+ (0.082688+ 0.108752 i) \ket{220}&
+ (0.139528+ 0.196955 i) \ket{221}&
+(- 0.112392- 0.109688 i ) \ket{222}.
\end{array}$}

\section{Conclusion}\label{sec:conclusion}
In this article, we explored the question of maximally entangled states, from the point of view of the evaluation of absolute values of invariants as a measure of entanglement, in the context of 3-qutrit states. We found new states that maximize the absolute value of the hyperdeterminant.  We showed that the Aharonov state is a simultaneous maximizer for the 3 fundamental invariants. We provided plots of the frequencies of the values of the invariants on semi-simple states, as well as level-set plots on the sphere, which illustrate their behavior. 

The idea of using algebraic invariants to measure entanglement has been investigated in the past in the case of qubit \cites{gour2014symmetric, chen2013proof} but not for qutrit. With the development of Noisy Intermediate Scale Quantum computers, one can ask the question of evaluating those algebraic invariants directly on a quantum device that produces quantum states. The evaluation of the Cayley hyperdeterminant on three-qubit states by means of measurement \cites{perez2020measuring, bataille2022quantum} opens the path to those similar questions for qutrit systems. 

In future works, it could be interesting to carry out similar analyses for more qubits or qudits. We anticipate some difficulty in the cases where a Jordan decomposition is not known or if there is not a generic semi-simple element that simplifies the computation and study of the invariants. Also, other challenges are anticipated since in most cases the entire invariant ring is not known, and the hyperdeterminant has high degree and can be difficult to compute. In terms of applications, it could be interesting to produce quantum protocols where the maximally entangled three-qutrit states found in this work exhibit different performances (like the Aharonov state and its role in protocols as explained in the introduction). In this respect we plan to investigate three-qutrit quantum games scenario based on Mermin's like inequalities \cite{Lawrence2017}.

\section*{Acknowledgements}
Holweck and Oeding acknowledge support from the Thomas Jefferson Foundation. Oeding was also supported by CNRS during part of this work.

\bibliography{biblio.bib}

\end{document}